\documentclass[11pt,a4paper]{article}

\usepackage{jheppub}
\usepackage{amsthm,mathtools}
\usepackage{tikz}
\usetikzlibrary{calc}
\usepackage{tikz-3dplot}
\usepackage{float}
\usepackage{placeins}
\usepackage[shortlabels]{enumitem}
\usepackage[nameinlink]{cleveref}
\hypersetup{hypertexnames=false}

\newcommand{\HMEC}{C^{\rm HMEC}}

\newcommand{\MMI}{\mathrm{MMI}}

\DeclareMathOperator{\spanop}{span}

\newtheorem{definition}{Definition}[section]
\newtheorem{lemma}[definition]{Lemma}

\newtheorem{theorem}[definition]{Theorem}

\newcommand{\piast}{{\Pi^\ast}}

\title{The Holographic Multi-Entropy Cone}
\author[a,b]{Xin-Xiang Ju,}
\author[b]{and Yang Zhao}

\emailAdd{juxinxiang21@mails.ucas.ac.cn}
\emailAdd{zhaoyang20a@mails.ucas.ac.cn}

\affiliation[a]{Institute for Advanced Study, Tsinghua University, Beijing 100084, China}
\affiliation[b]{School of Physical Sciences, University of Chinese Academy of Sciences, Zhongguancun east road 80, Beijing 100190, China}

\abstract{
We generalize the holographic entropy cone (HEC) to the holographic
multi-entropy cone (HMEC) by adjoining multi-entropy coordinates to
the standard entropy coordinates. We show that holographic states, through
their multi-entropy vectors, form a rational polyhedral cone in the multi-entropy
space, and multicontraction maps provide exact certificates for
holographic multi-entropy inequalities (HMEIs). We determine all facet inequalities of the
\(n=3,4\) HMECs, where \(n\) includes the purifier. These inequalities are organized into seven fundamental
HMEI orbits: two for \(n=3\) and five for \(n=4\).
We further propose two structural conjectures: HEC facet inequalities are
convex combinations of HMEC facet inequalities, and HMEC facet inequalities obey a
balanced-but-not-too-balanced principle.
}

\begin{document}

\maketitle

\section{Introduction}
\label{sec:introduction}
A basic problem in holography is to understand which quantum states can give rise to an emergent semiclassical bulk geometry. Entanglement provides a sharp entry point into this problem: the Ryu--Takayanagi formula identifies boundary entanglement entropy with the area of a bulk minimal surface~\cite{Ryu:2006bv,Ryu:2006ef}, while the Hubeny--Rangamani--Takayanagi prescription extends this relation to covariant settings~\cite{Hubeny:2007xt}. More generally, holographic entanglement entropy has become a central tool for probing how bulk geometry is encoded in boundary entanglement structures ~\cite{Rangamani:2016dms}, consistent with the idea that spacetime connectivity is built from entanglement~\cite{VanRaamsdonk:2010pw}. From this viewpoint, holographic entropy inequalities (HEIs) constrain the entanglement patterns compatible with a classical geometric dual.

The holographic entropy cone (HEC)~\cite{BaoEtAlHEC} provides the convex-geometric organization of these entropy constraints. For a fixed number of boundary subsystems, one records the entanglement entropies of all proper nonempty boundary subsets into a holographic entropy vector; the set of such vectors forms a rational polyhedral cone in entropy space. Its facets give the tightest linear HEIs, with monogamy of mutual information (MMI)~\cite{Hayden:2011ag} being the first genuinely holographic example.
Since the original graph model and the contraction map proof method, the HEC program has developed along several complementary directions:
the complete five-region cone and its extreme ray structure
\cite{HernandezCuenca:2019wgh,AvisHernandezCuencaFoundations}, the organization
of entropy relations and entropy arrangements
\cite{HubenyRangamaniRotaEntropyRelations,HubenyRangamaniRotaEntropyArrangement},
repackaged and balanced forms of HEIs
\cite{HeHeadrickHubenyRepackaged,HeHubenyRangamaniSuperbalance}, covariant and
time-dependent extensions and tests
\cite{CzechDongTimeDependentHEC,GradoWhiteGrimaldiHeadrickHubenyTesting,
GradoWhiteGrimaldiHeadrickHubenyMinimax}, and recent progress on infinite
families, contraction map classification, and six-region computations
\cite{CzechShuaiWangZhangEWNTopology,CzechLiuYuTwoInfiniteFamilies,
BaoNaskarContractionProperties,
BaoFuruyaNaskarClassification,CzechWangN7Inequality,
HernandezCuencaHubenyJiaHEIs,HeHubenyRotaN6}.
In this framework, once its facets are known, the HEC gives the complete set of linear
entropy constraints imposed by classical bulk geometry dual.

In this paper we extend the HEC to the holographic multi-entropy cone (HMEC)
by adjoining multi-entropy coordinates to the standard entropy vector. Multi-entropy is the
multipartite analogue of entanglement entropy. In quantum information, it is
obtained by analytically continuing R\'enyi multi-entropies, defined through
higher-dimensional toric-lattice contractions of density-matrix tensors, to the
limit in which the replica index approaches 1~\cite{GaddeKrishnaSharmaMultiEntropy}.
In holography, its proposed geometric dual is a minimal-area soap-film or
brane-web surface, closely related to minimal multi-cut prescriptions in
tensor networks
\cite{GaddeKrishnaSharmaMultiEntropy,GaddeKrishnaSharmaClassification,PeningtonWalterWitteveenFunReplicas,Hu:2026bhg}.
A key motivation for extending the HEC to the HMEC comes from recent results showing that multipartite entanglement measures can distinguish entanglement patterns characterized by Sperner hypergraphs, which are invisible to bipartition entropy data alone~\cite{GeneralizedBalance2602}.
Moreover, junction-law analyses show that multipartite measures are sensitive to
higher-codimension structures localized at junctions, providing diagnostics of
entanglement near contact regions in both gapped and gapless systems
~\cite{IizukaMiyataJunctionLaw,Iizuka:2026ahd}.

This motivation is reinforced from another direction by recent works on multipartite
entanglement signals \cite{BalasubramanianKangMurdiaRossSignals} and on the genuine multi-entropy \cite{IizukaGenuineMultiEntropy,Iizuka:2025caq}. In both lines of work, linear combinations of multi-entropies have emerged as increasingly important objects in holographic and quantum-information settings, and their nonvanishing is expected to encode meaningful information about the underlying entanglement structure.
  
The work~\cite{BalasubramanianKangMurdiaRossSignals} initiated a systematic
holographic study of multipartite
entanglement signals, whose nonvanishing serve as sufficient
witnesses for specific multipartite entanglement or correlation structures,
although they are not in general necessary conditions for their presence. \cite{BalasubramanianKangMurdiaRossSignals} analyzes known three- and four-party signals and proposes higher-party generalizations in AdS$_3$/CFT$_2$ and multiboundary wormhole states.
Stabilizer and graph-state
constructions provide controlled non-holographic laboratories in which
multi-invariants and multi-entropies can be computed explicitly
\cite{AkellaGaddePandeyStabilizer,Czech:2026cxw}. Other recent proposals have
constructed tripartite signals from entanglement wedge cross sections and
multipartite entanglement of purification
\cite{JuPanSunZhaoTriangleInformation,BaoFuruyaNaskarTripartiteSignal},
and related studies have analyzed extremal values of holographic entanglement
signals
\cite{JuPanSunZhaoHolographicMultipartiteUpperBound,JuPanSunWangZhaoMoreUpperBound,
JuSunZhaoUpperBoundCombinations,JuLiuSunXuZhaoIRModifiedGeometries}.  Meanwhile, the genuine multi-entropy was introduced in \cite{IizukaGenuineMultiEntropy} as a special kind of signal designed to extract genuine multipartite entanglement \cite{Iizuka:2025caq,Iizuka:2025elr}. A more structural mathematical formulation of genuine multipartite
signals has been developed using local-unitary invariants and M\"obius
inversion on the partition lattice \cite{GaddeGenuineSignals}. In holography, all these signals above are already subject to nontrivial structural
constraints: purely GHZ-like tripartite entanglement is forbidden in
time-symmetric holographic states, and further work has explored their time
evolution and mixed-state extensions
\cite{BalasubramanianKangCummingsMurdiaRossGHZ,
BalasubramanianJiangRossSignalDynamics,NaskarMixedStateSignalInequality}.
Together with recent work on four-party holographic constraints
~\cite{BalasubramanianEtAlFourParty}, and on multi-entropy-based diagnostics of
topological transitions and anisotropy in holographic semimetals
~\cite{Chen:2026hcf,Chen:2026xtx}, these results show that multipartite
entanglement signals are becoming a useful language for probing holographic
and quantum-information structures.

As such signals continue to proliferate, it becomes increasingly important to
understand their structural properties and organize them systematically.
Sign-definiteness plays a crucial role: the nonvanishing of a signal can
witness the presence of a particular multipartite structure, but if one wants
the vanishing of that signal to imply the absence of such a structure, or at
least to carry a meaningful structural interpretation, then the signal must
have a definite sign on the class of states under consideration. 
Holographic multi-entropy inequalities provide precisely this constraint
language: they characterize the sign-definite linear combinations of multi-entropies allowed by classical bulk geometry.

In this paper, we show that after adjoining partition-labeled multi-entropy
coordinates, the holographic multi-entropy cone \(\HMEC_n\) thus obtained again forms a rational polyhedral convex cone, with its projection onto bipartite entropy coordinates being
the standard HEC. We compute the complete cones \(\HMEC_3\) and
\(\HMEC_4\) in the total-label convention, where the purifier is included among
the \(n\) labels, and obtain seven fundamental holographic multi-entropy
inequalities. Using the complete \(n=3,4\) results together with partial \(n=5\)
evidence, we then propose two structural conjectures: standard HEC facets arise
from convex combinations of HMEC facets, and HMEC facet inequalities
obey a balanced-but-not-too-balanced principle.

The paper is organized as follows. Section~\ref{sec:framework} introduces the
multi-entropy vectors, the graph model description, polyhedrality, and the
multicontraction map proof method. Section~\ref{sec:low-dimensional} gives the
complete \(n=3,4\) facets, their projections to
subadditivity and MMI, the graph-realization proof of completeness for \(\HMEC_3\) and
\(\HMEC_4\), and the complexity analyses motivating our partial evaluation of
\(n=5\). Section~\ref{sec:conjectures} formulates the two guiding conjectures
and gives a finite certificate for the \(\HMEC_5\) high-balance obstruction.
Section~\ref{sec:conclusion} concludes.

\section{Multi-entropy, graph models, and multicontraction map certificates}
\label{sec:framework}

By generalizing the methods and tools of \cite{BaoEtAlHEC} to the multipartite setting, this section establishes the framework of the holographic multi-entropy cone. \ref{2.1}
defines the multi-entropy vectors obtained by adjoining multi-entropy coordinates. \ref{2.2} and \ref{2.3} build the equivalent graph model descriptions of bulk geometries. \ref{subsec:universal-graph-polyhedrality} reveals the rational polyhedral convex cone structure of HMEC, while \ref{2.5} verifies the
multicontraction map method used to prove holographic multi-entropy
inequalities. These ingredients are parallel to their counterparts in HEC, and lay the foundation of the low-$n$ HMEC facet computation
in Section~\ref{sec:low-dimensional}.

\subsection{Multi-entropy vectors and HMEC}\label{2.1}

Multi-entropy \(S^{(q)}\) is the multipartite analogue of entanglement entropy. In quantum
information, it is obtained from R\'enyi multi-entropies defined by
higher-dimensional toric-lattice contractions of density-matrix tensors,
followed by analytic continuation of the replica index to one. For \(q=2\), it
reduces to the standard von Neumann entropy
\cite{GaddeKrishnaSharmaMultiEntropy}. 

In holographic theories, the proposed
dual of \(S^{(q)}\) is the area of a minimal brane web, or soap-film, anchored
on the corresponding boundary partition \cite{GaddeKrishnaSharmaMultiEntropy}. For \(q\ge3\), this multiway soap-film divides the bulk into chambers homologous to the boundary subsystems, while for
\(q=2\) it reduces to the usual RT surface in the bulk separating a boundary
region from its complement. 
In the following discussion, we will refer to this soap-film for any $q$ as the minimal multi-entropy surface.

We now define the holographic multi-entropy vector. We use a total-label
convention: The complete boundary system is divided into \(n\) labeled regions
\(A_1,\dots,A_n\), which includes the ``purifier" region, unlike in the standard HEC notations. Consider each
nontrivial partition\footnote{Here $\piast([n])$ denotes the nontrivial partitions of $[n]$, where the trivial partition that consists of only one block containing all $n$ elements is excluded. {Note that here for simplicity, $[n]$ is used as an abbreviation for boundary subregions $A_1,\dots ,A_n$, and $\pi$ in fact partitions these $n$ subregions.}}
\(\pi=B_1:\cdots:B_{|\pi|}\in\piast([n])\), where \([n]=\{1,\ldots,n\}\), \(|\pi|\) is the number of blocks of \(\pi\), and each
block \(B_i\) contains one or more boundary regions. We define the corresponding
multi-entropy coordinate\footnote{When \(q=2\), this gives the standard
entanglement entropy. In particular, for a bipartition \(A:A^c\), with the
complement taken inside the $n$ parties, we write \(S_A=S^{(2)}(A:A^c)\).}
\(S^{({|\pi|})}(\pi):=S^{({|\pi|})}(B_1:\cdots:B_{|\pi|})\). By collecting the values of such
multi-entropies for each partition in \(\piast([n])\), we obtain the
holographic multi-entropy vector 
\begin{equation}
    \vec S:=\left(S^{(|\pi|)}(\pi)\right)_{\pi\in\Pi^*([n])},
\end{equation}
with length \(|\Pi^*([n])|=B_n-1\), where \(B_n\)
is the \(n\)-th Bell number.

Equivalently, for fixed $n$, $\vec S$ is obtained by starting from the standard entropy vector, which consists of the entropy coordinates associated with all bipartitions of the full boundary system into a subsystem and its complement, and adjoining the multi-entropy coordinates labeled by all partitions with three or more blocks.
 For example, for \(n=3\) the multi-entropy vector is
\((S_A,S_B,S_C,S^{(3)}(A:B:C))\). For \(n=4\), it is
\begin{align}
\big(&S_A,S_B,S_C,S_D;\; S_{AB},S_{AC},S_{AD}; \notag\\
&S^{(3)}(AB:C:D),S^{(3)}(AC:B:D),S^{(3)}(AD:B:C), \notag\\
&S^{(3)}(A:BC:D),S^{(3)}(A:BD:C),S^{(3)}(A:B:CD); \notag\\
&S^{(4)}(A:B:C:D)\big).
\end{align}
Meanwhile, for \(n=2\) it reduces to the entropy vector \((S_A,S_B)\).

For fixed \(n\), varying over boundary decompositions and admissible bulk
geometries gives a convex cone, which we denote by \(C^{\mathrm{HMEC}}_n\). The convex cone structure follows by the same argument as for HEC: rescaling the manifold metric rescales the multi-entropy vector, and
taking the disjoint union of manifolds adds multi-entropy vectors coordinate-wise. Thus,
\(C^{\mathrm{HMEC}}_n\) is closed under positive scalar multiplication and
addition. In Section~\ref{subsec:universal-graph-polyhedrality} we show that
\(C^{\mathrm{HMEC}}_n\) is a rational polyhedral cone, so it is bounded by
finitely many holographic multi-entropy inequalities as its facets.

\subsection{Graph model description}\label{2.2}

The continuum soap-film picture discussed above has a natural discrete
counterpart in weighted graph models. Given any bulk geometry, we now show the construction of its corresponding graph. 

A graph model \(G\) consists of boundary
vertices labeled by \([n]\), internal vertices, and nonnegative edge weights.
The guiding idea is the same as in HEC:
minimal separating surfaces in the bulk are replaced by minimal cuts in a
weighted graph. For multi-entropy, however, the relevant cuts are generalized from bipartite ones to multiway cuts associated with the multi-entropies. 

For a bulk geometry whose boundary is decomposed into regions
\(A_1,\dots,A_n\), draw in the bulk the minimal multi-entropy surface that gives $S^{({|\pi|})}(\pi)$ for each
nontrivial partition \(\pi\in \piast([n])\). These
multi-entropy surfaces collectively partition the bulk into finitely many
chambers. Similar to \cite{BaoEtAlHEC}, a vertex is assigned to each connected
chamber. Chambers adjacent to the asymptotic boundary define boundary vertices,
which are colored\footnote{Identical to the HEC case, coloring a vertex subset with another set means creating a function mapping the former to the latter.} by the corresponding boundary region.

An edge is assigned to each pair of chambers that share a nonzero-area
interface. Its weight is the area of the corresponding interface divided by
\(4G_N\). Thus, we obtain the weighted graph model \(G=(V,E)\) for the given
geometry. See panel (a) of Figure~\ref{fig:graphmodel-dictionary} for an example.

Moreover, parallel to the minimal multi-entropy surfaces, we introduce the notion of multiway cuts as generalizations of the cuts in \cite{BaoEtAlHEC}. Each multiway cut $W=(W_1:\dots: W_q)$ is a nontrivial partition of the graph vertices $V$ defined above. Also, for every such multiway cut, let $C(W)$ be the collection of edges whose two endpoints reside in different blocks of $W$, and $|C(W)|$ counts the total weights of edges in $C(W)$.

We can now define the discrete multi-entropy. For any \(\pi\in\piast([n])\), set
\begin{equation}
    S^{(|\pi|)\ast}(\pi):=\min_{\substack{W\in\piast(V)\\ \partial W \sim \pi}}|C(W)|.
\end{equation}
This is the minimum multiway-cut cost separating the boundary blocks of
\(\pi\). Here \(\partial W\sim\pi\) means that the minimization is over
multiway cuts whose induced partition on the boundary vertices agrees with the
boundary partition \(\pi\).

\begin{figure}[h]
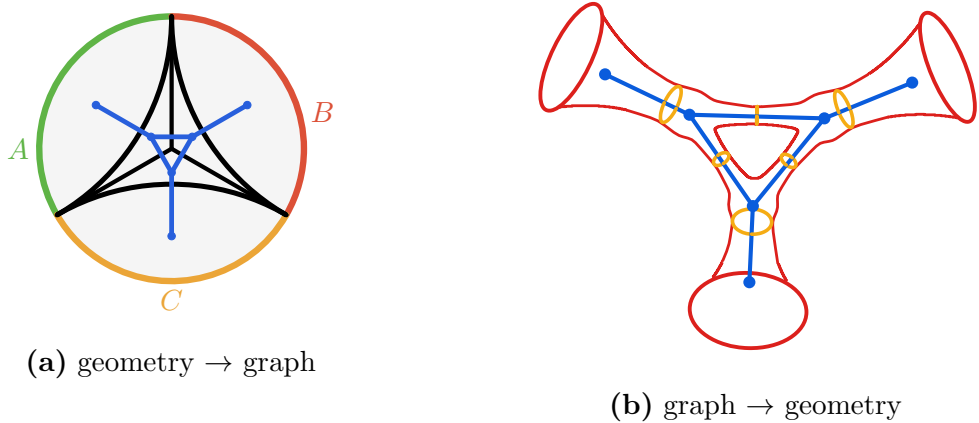

\centering
\begin{minipage}[c]{0.48\textwidth}
\centering


\vspace{0.4em}
\textbf{(b)} graph $\rightarrow$ geometry
\end{minipage}
\caption{Geometry--graph correspondence for multi-entropy. Panel (a)
illustrates the geometry-to-graph construction: the boundary is divided into the parties
\(A,B,C\), the outer black curves are RT surfaces, and the central trivalent
network is the minimal multi-entropy surface for \(S^{(3)}(A:B:C)\). Denoted in blue is the graph corresponding to this geometry setting. Panel (b) illustrates the
graph-to-geometry construction: graph edges are replaced by wormhole necks and
graph vertices by high-cost chambers. The solid yellow curves indicate wormhole throat
cross sections. By choosing the chamber costs to be sufficiently large, all
minimizing multi-entropy surfaces localize on the throats, thereby reproducing
the graph multiway-cut values.}
\label{fig:graphmodel-dictionary}
\end{figure}

There is, however, an important difference between the graph model for entropy vectors and its multi-entropy counterpart. Since entropy coordinates correspond to bipartite cuts, discrete entropies can be computed in polynomial time by max-flow/min-cut algorithms
\cite{FordFulkerson1962}. By contrast, discrete multi-entropies are given by
minimum multiway cuts. Computing these cuts is already NP-hard in general
weighted graphs with three boundary vertices \cite{Dahlhaus:1994mtc}. Therefore one should not
expect a universal polynomial-time local simplification rule that preserves all
multi-entropy coordinates and makes every multiway cut
readable from a simple reduced graph; such a rule would give a polynomial-time
algorithm for the three-terminal multiway-cut problem. This is
the reason why graph realization for HMEC is algorithmically harder than for
standard HEC.

\subsection{Graph to geometry}\label{2.3}
We now describe the opposite graph-to-geometry direction. Given a finite weighted graph
model, we construct a multi-mouth wormhole geometry whose multi-entropy values reproduce
the corresponding minimum multiway-cut values of the discrete multi-entropies. Together with the
geometry-to-graph construction in the previous subsection, this gives the equivalent graph model formulation used
below to establish polyhedrality and to certify HMEIs.

To be precise, let \(G=(V,E,w)\) be a finite weighted graph with boundary labels \([n]\), with $w$ denoting the edge weights.
For a nontrivial partition \(\pi=B_1:\cdots:B_{|\pi|}\in \Pi^*([n])\), let $S_G^{(|\pi|)\ast}(\pi)$ be its associated discrete multi-entropy, and we claim that \(G\) can be realized by a multi-mouth wormhole geometry \(X_G\) with negative curvature such that
\begin{equation}
    S_{X_G}^{(|\pi|)}(\pi)
=
S_G^{(|\pi|)\ast}(\pi)
\qquad
\forall \pi\in \Pi^*([n]),
\end{equation}
where \(S_{X_G}^{(|\pi|)}(\pi)\) denotes the multi-entropy in \(X_G\),
computed as the minimum soap-film area divided by \(4G_N\).  See panel (b) of Figure~\ref{fig:graphmodel-dictionary} for an example.

The construction replaces every edge \(e=(u,v)\) by a wormhole neck \(N_e\)
connecting the geometric pieces associated with vertices \(u\) and \(v\). The
neck contains a distinguished throat
\(\Gamma_e\), whose area is fixed by the edge weight:
\begin{equation}
\frac{\operatorname{Area}(\Gamma_e)}{4G_N}
=
w_e .
\end{equation}
In a two-dimensional spatial slice this area is simply the length of a closed
geodesic. Meanwhile, every graph vertex \(v\) of degree \(d(v)\) (which is the number of edges incident to $v$) is replaced
by a negatively curved \(d(v)\)-holed pants chamber \(P_v\): for \(d(v)=3\)
this is a standard hyperbolic pair of pants. The neck-facing boundaries of \(P_v\) are glued to the wormhole necks for the edges attached to \(v\). For boundary vertices, their corresponding pants chambers also carry the corresponding external asymptotic boundary component.

The metric is chosen to satisfy two elementary conditions. First,
each neck \(N_e\) satisfies a bottleneck condition: any admissible interface
separating the two ends of \(N_e\) has area at least
\(\operatorname{Area}(\Gamma_e)\), with equality achieved by the throat
\(\Gamma_e\). Second, the vertex chambers satisfy a high-cost condition: every multiway web lying inside a chamber \(P_v\) that nontrivially partitions the chamber holes has area
larger than the total throat area of all the wormhole necks that it connects. This condition forces minimizing webs to localize
on neck throats rather than residing in the partition vertex chambers.

With the bottleneck and high-cost chamber conditions imposed, for every boundary partition \(\pi\), the minimizing web that gives the corresponding multi-entropy can be
pushed to a union of wormhole neck throats without increasing its area, hence the continuum minimization reduces to the graph multiway-cut minimization. In this way, the constructed
manifold realizes all discrete multi-entropies of the original graph. Thus, for any finite weighted graph, we can always construct a geometry where each multi-entropy coordinate equals the discrete multi-entropy in the graph. Combined with the result in the previous subsection that an arbitrary bulk geometry admits a graph model description, we can conclude the equivalence between bulk geometries and graph models. Under this equivalence, crucial properties of general HMEC are sufficiently proved using the graph model formalizationin the following two subsections.

\subsection{Universal graph language and polyhedrality}
\label{subsec:universal-graph-polyhedrality}

In this section, we prove that \(C_n^{\rm HMEC}\) is a rational polyhedral cone
for any fixed \(n\). Equivalently, it is cut out by finitely many
facet-defining HMEIs. As in \cite{BaoEtAlHEC}, this can be shown by converting
the arbitrary graph models that characterize holographic multi-entropy vectors to a universal graph description.

\begin{lemma}\label{universalgraph}
Any holographic multi-entropy vector in $C_n^{\rm HMEC}$ admits a universal complete graph model description with fixed boundary coloring.
\end{lemma}
\begin{proof}
To define the vertex set \(V_U\) of the universal graph \(G_U\), we generalize
the bitstring description in \cite{BaoEtAlHEC} to label-strings in our
multi-entropy setting. For any
\(\pi=B_1:\cdots:B_{|\pi|}\in\piast([n])\), we introduce a label set
\(\Lambda_\pi:=\{\lambda^\pi_1,\dots,\lambda^\pi_{|\pi|}\}\) that assigns one
label to each block of $\pi$. By choosing a label from $\Lambda_\pi$ for all $\pi\in\piast\big([n]\big)$ and combining them together, a label-string of length $|\piast\big([n]\big)|$ is formed. Each distinct label-string of
this form represents a universal vertex, hence
\(|V_U|=\prod_{\pi\in\piast([n])}|\pi|\). {These vertices encode all possible
simultaneous block-membership patterns for the partitions in \(\Pi^*([n])\),
regardless of whether a given pattern is realized in any specific geometry.}
In particular, any boundary vertex $x_i$ with $i\in[n]$ is defined such that $(x_i)_\pi=\lambda^\pi_p$ iff $i\in B_p$, where $p\in\big[|\pi|\big]$.

As for the universal edge set \(E_U\), for each \(x\in V_U\) we introduce \(W(x):=\cap_{\pi\in\piast\big([n]\big)}W_\pi^{x_\pi}\), where \(W_\pi\) is chosen as the multiway cut in the original graph that yields the discrete multi-entropy \(S^{(|\pi|)\ast}\). For any \(p\in\big[|\pi|\big]\), \(W_\pi^{\lambda^\pi_p}\)  is defined as the block of \(W_\pi\) in which boundary vertices are colored by \(B_p\). {Thus, given an arbitrary graph model \(G\), \(W(x)\) is the cell consisting of the vertices whose residing block in each multiway cut is exactly specified by \(x\).} Analogously to the HEC case, some of these cells can be disconnected or empty, and \(W(x_i)\) contains precisely the boundary vertices in \(V(G)\) colored by \(i\).
The universal edge weights are then defined as in \cite{BaoEtAlHEC}:
\(w(x,y):=\sum_{e\in E(x,y)}w_e\), where \(E(x,y)\subset E(G)\) contains
edges in the original graph with one endpoint in \(W(x)\) and the other in
\(W(y)\).

It is then straightforward that for any $\pi\in \piast([n])$, the discrete multi-entropy in the universal complete graph equals that of the original graph. 
\end{proof}

Based on Lemma \ref{universalgraph} and the graph-geometry dictionary from previous subsections, all holographic multi-entropy vectors in \(C_n^{\rm HMEC}\) can be represented by the universal complete graph \(G_U=(V_U,E_U)\) with fixed boundary coloring, while the only difference lies in the non-negative weights of the edges. Therefore, the proof of polyhedrality is formally identical to the argument for the holographic entropy cone, only with cuts replaced by multiway cuts: For each \(\pi\in\Pi^\ast([n])\), there are finitely many feasible multiway cut assignments, each with cost linear in the edge weights \(w\). As \(\pi\) ranges over \(\Pi^\ast([n])\), equal-cost hyperplanes between these assignments divide the weight orthant into a finite number of rational polyhedral subcones. On each of them the minimizing assignment for every \(\pi\) is fixed, and the full multi-entropy vector is linear in \(w\). The image is therefore a finite union of rational polyhedral cones and is convex by construction; equivalently, it is the conical hull of finitely many rational polyhedral images. Therefore \(C_n^{\rm HMEC}\) is rational polyhedral.

% \begin{itemize}[leftmargin=*]
%   \item For fixed \(n\), multiway-cut behavior is controlled by finitely many partition-assignment profiles.  This converts arbitrary graph models into a finite universal graph language.

%   \item The planned polyhedrality statement is
%   \[
%     C_n^{\rm HMEC}\text{ is a rational polyhedral cone for fixed }n.
%   \]
%   The proof should be sketched through the universal graph construction rather than through case-by-case geometries.

%   \item Polyhedrality is essential because it implies finitely many HMEC facets at fixed \(n\).  It turns the comparison between HEC and HMEC into a finite convex-geometric problem.

%   \item The exposition should avoid making polyhedrality look tautological.  The cone is defined using arbitrary weighted graphs, and the finite universal certificate is the nontrivial step.
% \end{itemize}

\subsection{The multicontraction map proof method}\label{2.5}

The contraction map proof method, first proposed in \cite{BaoEtAlHEC}, is a crucial tool for proving HEIs. In this section, we generalize it to the HMEC case, thereby enabling a systematic and efficient search for valid HMEIs.

After moving all terms with negative coefficients to the right-hand side, an HMEI can be written as
\begin{equation}
\sum_{l\in L} a_l S^{(|\pi_l|)}(\pi_l)
\ge
\sum_{r\in R} b_r S^{(|\rho_r|)}(\rho_r),
\label{eq:hmec-inequality}
\end{equation}
where for any $l\in L, r\in R$, $a_l, b_r$ are positive, and
$\pi_l,\rho_r\in\piast([n])$. Analogously to \cite{BaoEtAlHEC}, the central
idea of the multicontraction map proof is to construct feasible candidates for
the right-hand-side multi-entropy surfaces by recombining {segments of} the minimal multi-entropy surfaces for the left-hand-side terms. We
therefore introduce separate label-string spaces for the left- and right-hand
sides of the inequality, rather than using the universal label-strings that
encode all multi-entropy coordinates at once. For each
$\pi=B_1:\cdots:B_{|\pi|}\in\piast([n])$ with label set
$\Lambda_\pi:=\{\lambda^\pi_1,\dots,\lambda^\pi_{|\pi|}\}$, define
\begin{equation}
X_L:=\prod_{l\in L}\Lambda_{\pi_l},
\qquad
X_R:=\prod_{r\in R}\Lambda_{\rho_r},
\end{equation} 
as extensions of the bitstring spaces $\{0,1\}^L$ and $\{0,1\}^R$ in \cite{BaoEtAlHEC}.
Thus, an element $x\in X_L$ specifies one block of $\pi_l$ for each left-hand-side term, and similarly for $X_R$.

With these spaces fixed, for each boundary color $i\in[n]$, the occurrence
vectors $x_i\in X_L$, $y_i\in X_R$ are now defined by
\begin{equation}
(x_i)_l:=\lambda^{\pi_l}_m \ \text{iff}\ i\in B^{\pi_l}_m,\qquad (y_i)_r:=\lambda^{\rho_r}_p \ \text{iff}\ i\in B^{\rho_r}_p
\end{equation}
for each $m\in\big[|\pi_l|\big]$ and $p\in\big[|\rho_r|\big]$, and we write $B^{\pi_l}_m,B^{\rho_r}_p$ as the respective blocks of $\pi_l,\rho_r$ for clarity.
Furthermore, we need to equip $X_L,X_R$ with the correct modified weighted Hamming distances
\begin{equation}
d_L(x,x')=
\sum_{l\in L}a_l \mathbf 1_{x_l\ne x'_l},
\qquad
d_R(y,y')=
\sum_{r\in R}b_r\mathbf 1_{y_r\ne y'_r},
\end{equation}
where $\mathbf 1_{x_l\ne x'_l}$ equals one if $x_l\ne x'_l$ and zero
otherwise. The indicator $\mathbf 1_{y_r\ne y'_r}$ is defined in the same way.

Based on the above concepts, we can now define the multicontraction map as
follows:
\begin{definition}[\normalfont multicontraction map]\footnote{When every partition has two blocks, this definition reduces to the contraction map proof method for HEIs in \cite{BaoEtAlHEC}.}
For any multi-entropy inequality in the form \eqref{eq:hmec-inequality}, a
multicontraction map is a function $f:X_L\rightarrow X_R$ that satisfies
$f(x_i)=y_i$ for all $i\in[n]$, and
\begin{equation}
d_R\big(f(x),f(x')\big)\le d_L(x,x'),
\qquad \forall x,x'\in X_L.
\label{eq:multi-contraction-condition}
\end{equation}
\end{definition}

\begin{theorem}\label{thm:multi-contraction}
If there exists a multicontraction map for  \eqref{eq:hmec-inequality}, then the inequality holds for every graph model. Hence it defines a valid HMEI.
\end{theorem}

\begin{proof}
The proof of Theorem \ref{thm:multi-contraction} is the direct multipartite analogue of the contraction map proof for HEC: cuts are replaced by multiway cuts, bitstrings by label-strings, and the weighted Hamming distance by its counterpart on label-string spaces. 

Consider an arbitrary graph model \(G\). For each left-hand-side partition
\(\pi_l=B^{\pi_l}_1:\cdots:B^{\pi_l}_{|\pi_l|}\), choose a minimum multiway
cut \(W_{\pi_l}\) realizing the discrete multi-entropy
\(S_G^{(|\pi_l|)\ast}(\pi_l)\). For any $p$ from $1$ to $|\pi_l|$, write
\(W_{\pi_l}^{\lambda^{\pi_l}_p}\) for the block of this cut associated with
\(B^{\pi_l}_p\). Then every graph vertex \(v\in V(G)\) determines a label
string \(x_L(v)\in X_L\) by recording its block in each chosen left-hand-side
cut:
\begin{equation}
x_L(v)_l=\lambda^{\pi_l}_p
\longleftrightarrow
v\in W_{\pi_l}^{\lambda^{\pi_l}_p}.
\end{equation}
The collection of such label-strings records the simultaneous block-membership
patterns that are realized in
\(G\), with respect to all chosen left-hand-side cuts. Distinct vertices may of course have the same label-string\footnote{Here we simplify the discussion by working directly with the encoding label-string $x_L(v)$.
Equivalently, for any \(x\in X_L\) define
\begin{equation}
W(x)=\bigcap_{l\in L} W_{\pi_l}^{x_l}.
\end{equation}
They are the common-refinement cells of the
left-hand-side multiway cuts, among which the nonempty ones give the graph model analogues of the
bulk pieces cut out by the left-hand-side minimal multi-entropy surfaces. For
fixed \(r\in R\), the right-hand-side cut constructed from the
multicontraction map has blocks
\[
\bigcup_{x:\, f(x)_r=\lambda^{\rho_r}_p} W(x).
\]
Thus the proof recombines the cells determined by the left-hand-side cuts into
feasible right-hand-side multiway cuts, directly paralleling the contraction
map argument of \cite{BaoEtAlHEC}.
}.

By the multicontraction map $f$, the above label-strings $x_L(v)$ for all
$v\in V(G)$ map to \(X_R\). These images define a feasible cut for every
right-hand-side partition. Specifically, for each $r\in R$, define a partition
$W'_r\in\piast\big(V(G)\big)$ by grouping all vertices with the same $r$-th
component $f\big(x_L(v)\big)_r$ into one block. The boundary condition
$f(x_i)=y_i$ ensures that, for fixed $r$ and $p$, one has
$f(x_i)_r=\lambda_p^{\rho_r}$ if and only if
$i\in B_p^{\rho_r}$. Hence the induced partition $W'_r$ separates the boundary
vertices according to
\(\rho_r=B^{\rho_r}_1:\cdots:B^{\rho_r}_{|\rho_r|}\), and is therefore a
feasible multiway cut for $\rho_r$.
Thus, by definition of the discrete multi-entropy,
\begin{equation}
S_G^{(|\rho_r|)\ast}(\rho_r)
\le|C(W_r')|=
\sum_{(u,v)\in E(G)}
w_{uv}\ \mathbf 1_{f\big(x_L(u)\big)_r\ne f\big(x_L(v)\big)_r},
\end{equation}
where $w_{uv}$ is the weight of the edge $(u,v)$.
Summing over $r\in R$ with coefficients $b_r$, we obtain
\begin{equation}
\begin{aligned}
\sum_{r\in R}b_rS_G^{(|\rho_r|)\ast}(\rho_r)
\le\sum_{r\in R}b_r|C(W_r')|=
\sum_{(u,v)\in E(G)}&
w_{uv}\sum_{r\in R}b_r\mathbf 1_{f\big(x_L(u)\big)_r\ne f\big(x_L(v)\big)_r}
\\=
\sum_{(u,v)\in E(G)}w_{uv}\ 
d_R\big(f(x_L(u)), f(x_L(v))\big)
&\le
\sum_{(u,v)\in E(G)}w_{uv}\ 
d_L\big(x_L(u),x_L(v)\big),
\end{aligned}
\end{equation}
where the last inequality follows from the contraction condition
\eqref{eq:multi-contraction-condition}. Meanwhile, on the left-hand side of (\ref{eq:hmec-inequality}), 
\begin{equation}
    \begin{aligned}
    \sum_{(u,v)\in E(G)}w_{uv}\ d_L\big(x_L(u),x_L(v)\big)= \sum_{(u,v)\in E(G)}w_{uv}
    \sum_{l\in L}a_l\ \mathbf 1_{x_L(u)_l\ne x_L(v)_l}
    \\ =\sum_{l\in L}a_l\sum_{(u,v)\in E(G)}w_{uv}
    \ \mathbf 1_{x_L(u)_l\ne x_L(v)_l}=\sum_{l\in L}a_l|C(W_{\pi_l})|=\sum_{l\in L}a_l S_G^{(|\pi_l|)\ast}(\pi_l).
    \end{aligned}
\end{equation}
Thus, we have proved that the HMEI
\begin{equation}
\sum_{l\in L} a_l S^{(|\pi_l|)}(\pi_l)
\ge
\sum_{r\in R} b_r S^{(|\rho_r|)}(\rho_r),
\end{equation}
holds for every graph model. Under the graph-geometry dictionary, this HMEI is generally valid for any holographic state, and the multicontraction map serves as a general tool of proving HMEIs and identifying HMECs for a fixed number of boundary parties.
\end{proof}

We record two basic HMEIs and their multicontraction maps as examples:
\begin{align}
    &2S^{(3)}(A:B:C)
    \ge
    S^{(2)}(A:BC)+S^{(2)}(B:AC)+S^{(2)}(C:AB),
    \label{eq:basic-three-party-hmei}
    \\
    &S^{(4)}(A:B:C:D)+S^{(2)}(AD:BC)
    \ge
    S^{(3)}(AD:B:C)+S^{(3)}(A:BC:D).
    \label{eq:orbit-two-example}
\end{align}
Tables~\ref{tab:basic-three-party-map} and~\ref{tab:orbit-two-map} give multicontraction maps
satisfying the boundary and contraction conditions. For clarity we set the labels of $\Lambda_\pi$ to be the blocks of $\pi$.

\begin{table}[h]
\centering
\begin{tabular}{c|ccc}
\hline
\(2S^{(3)}(A:B:C)\)
& \(S^{(2)}(A:BC)\)
& \(S^{(2)}(B:AC)\)
& \(S^{(2)}(C:AB)\) \\
\hline
\(A\) & \(A\) & \(AC\) & \(AB\) \\
\(B\) & \(BC\) & \(B\) & \(AB\) \\
\(C\) & \(BC\) & \(AC\) & \(C\) \\
\hline
\end{tabular}
\caption{Multicontraction map for
(\ref{eq:basic-three-party-hmei}). }
\label{tab:basic-three-party-map}
\end{table}

\begin{table}[h]
\centering
\begin{tabular}{cc|cc}
\hline
\(S^{(4)}(A:B:C:D)\)
& \(S^{(2)}(AD:BC)\)
& \(S^{(3)}(AD:B:C)\)
& \(S^{(3)}(A:BC:D)\) \\
\hline
\(A\) & \(AD\) & \(AD\) & \(A\)  \\
\(A\) & \(BC\) & \(AD\) & \(BC\) \\
\(B\) & \(AD\) & \(AD\) & \(BC\) \\
\(B\) & \(BC\) & \(B\)  & \(BC\) \\
\(C\) & \(AD\) & \(AD\) & \(BC\) \\
\(C\) & \(BC\) & \(C\)  & \(BC\) \\
\(D\) & \(AD\) & \(AD\) & \(D\)  \\
\(D\) & \(BC\) & \(AD\) & \(BC\) \\
\hline
\end{tabular}
\caption{Multicontraction map for (\ref{eq:orbit-two-example}).}
\label{tab:orbit-two-map}
\end{table}

\FloatBarrier

\section{Complete HMEC facets for $n=3,4$}
\label{sec:low-dimensional}

In the last section, we proved that \(\HMEC_n\) is polyhedral and
introduced multicontraction maps as certificates for holographic multi-entropy
inequalities. We now apply this framework to \(\HMEC_3\) and \(\HMEC_4\).
As a result, the complete 7 facet orbits of \(\HMEC_3\) and \(\HMEC_4\) have been determined: \(2\) in \(\HMEC_3\) and \(5\) in \(\HMEC_4\). Not only do they serve as basic examples of HMEC facet computation and reveal low-dimensional holographic multipartite entanglement structures, but by comparing them with projections from the entropy-sector, such as subadditivity and MMI, the results also shed light on some general properties of HMEC.

\subsection{The \texorpdfstring{\(n=3\)}{n=3} cone and the subadditivity projection}

Fixing the total-label convention, we first compute \(\HMEC_3\), and then compare it with the
standard HEC with two boundary subregions and a purifier subregion. We label the boundary subregions by $A,B$ and $C$.

In \(\HMEC_3\), there are \(2\) orbits. Here we define an orbit as the set consisting of an HMEI and all inequalities obtained from it by permuting the boundary regions. The orbit size is the number of inequalities in the orbit. The first orbit contains
\begin{equation}
P_A:\qquad
-
S^{(3)}(A:B:C)
+
S_B
+
S_C
\ge 0 ,
\end{equation}
together with its permutations, and has orbit size \(3\). The second orbit is
the symmetric inequality
\begin{equation}
K_{ABC}:\qquad
2S^{(3)}(A:B:C)
-
S_A
-
S_B
-
S_C
\ge 0 ,
\end{equation}
and has orbit size \(1\). {In fact, similar inequalities were proposed in \cite{Bao:2018gck} in the context of multipartite entanglement of purification, which was identified as the reflected multi-entropy in \cite{Yuan:2024yfg}. In the special case where the boundary parties considered in \cite{Bao:2018gck} form a pure state, the relevant bulk geometric object reduces to the minimal multi-entropy surface, and the corresponding inequalities become $P_A$ and $K_{ABC}$.
}

To show that the inequalities in both orbits give the complete set of facets of the HMEC requires two steps. First, a multicontraction map corresponding to each orbit have already been constructed, thereby proving that they give valid HMEIs. (For $P_A$ this is direct, while for $K_{ABC}$ see Table \ref{tab:basic-three-party-map}.) Second, we show that these inequalities are tight by computing the extreme rays of the cone bounded by these inequalities and giving a graph realization for each ray.
Writing the coordinate order as
\begin{equation}
(x_1,x_2,x_3,x_4)
=
\bigl(
S_A,\,
S_B,\,
S_C,\,
S^{(3)}(A:B:C)
\bigr),
\end{equation}
there are \(4\) concrete extreme rays in the cone formed by the two orbits:
\begin{equation}R_1:
(1,1,1,2),
\qquad
R_2:
(0,1,1,1),
\qquad
R_3:
(1,0,1,1),
\qquad
R_4:
(1,1,0,1).
\end{equation}
All of them admit simple tree graph realizations: \(R_1\) corresponds to a graph in which boundary vertices \(A,B,C\) are all connected to a single inner point by weight-one edges, while \(R_2\) corresponds to a graph where only one weight-one edge exists, connecting \(B\) and \(C\). The graphs for $R_3$ and $R_4$ are just permutations of that of $R_2$.
Thus, these two orbits constitute all \(4\) facets of
\(\HMEC_3\). 

{By the equivalence between the bulk geometry and the graph-model description, any inequality satisfied by all finite weighted graph models is also satisfied in any time-reflection-symmetric holographic spacetime. Hence the inequalities found above give the complete set of facets of \(\HMEC_3\). The same reasoning applies to the \(\HMEC_4\) case discussed in the next subsection.
}

Meanwhile, for the HEC, the only nontrivial facets are given by subadditivity,
\begin{equation}
S_A+S_B-S_C\ge 0,
\end{equation}
and its permutations.
It can be seen that in \(\HMEC_3\), this entropy inequality is resolved into a positive conic combination of
\(\HMEC_3\) facets:
\begin{equation}
2P_C+K_{ABC}
=
S_A+S_B-S_C .
\end{equation}
After fixing a common normalization, it becomes a convex
combination. 

This is the simplest example of how an entropy-only HEI can be refined by facets of HMEC.

\subsection{The \texorpdfstring{\(n=4\)}{n=4} cone and the MMI projection}\label{3.2}

When $n=4$, the exact \(\HMEC_4\) computation gives \(27\) concrete facets. They are
organized into 7 orbits up to \(S_4\) symmetry, that is, up to the permutation of $A,B,C$ and $D$.  \(2\) of the \(7\) orbits are lifts from lower-party HMEIs, and their representatives are obtained by applying the
two \(\HMEC_3\) facets to the effective three-block grouping \((AB),C,D\):
\begin{equation}
-
S^{(3)}(AB:C:D)
+
S_C
+
S_D
\ge 0 ,
\end{equation}
and
\begin{equation}
2S^{(3)}(AB:C:D)
-
S_{AB}
-
S_C
-
S_D
\ge 0 .
\end{equation}
The five remaining orbits give genuinely new inequalities involving four-partitite multi-entropies.  The representatives are
\begin{equation}
\begin{aligned}
F_{4.1}:&\quad
S^{(4)}(A:B:C:D)
-
S^{(3)}(AD:B:C)
-
S^{(3)}(A:BC:D)
+
S_{AD}
\ge 0 ,
\\[0.8ex]
F_{4.2}:&\quad
-
S^{(4)}(A:B:C:D)
+
S^{(3)}(AD:B:C)
+
S^{(3)}(A:BD:C)
\\
&\quad
+
S^{(3)}(A:B:CD)
-
S_A
-
S_B
-
S_C
\ge 0 ,
\\[0.8ex]
F_{4.3}:&\quad
-
S^{(4)}(A:B:C:D)
+
S^{(3)}(AD:B:C)
+
S^{(3)}(A:BC:D)
\\
&\quad
+
S_{AB}
+
S_{AC}
-
S_A
-
S_B
-
S_C
-
S_D
\ge 0 ,
\\[0.8ex]
F_{4.4}:&\quad
2S^{(4)}(A:B:C:D)
-
S_{AB}
-
S_{AC}
-
S_{AD}
\ge 0 ,
\\[0.8ex]
F_{4.5}:&\quad
3S^{(4)}(A:B:C:D)
-
S^{(3)}(AB:C:D)
-
S^{(3)}(AC:B:D)
\\
&\quad
-
S^{(3)}(AD:B:C)
-
S_A
\ge 0 .
\end{aligned}
\end{equation}

All these inequalities above
are certified by multicontraction maps, so they are valid for every
holographic graph model, and thus for every holographic state.
For extreme rays, note that the multi-entropy vector space has dimension \(14\).
The coordinate order is set as:
\begin{align}
(x_1,\ldots,x_{14})=
\big(&S_A,S_B,S_C,S_D;\; S_{AB},S_{AC},S_{AD}; \notag\\
&S^{(3)}(AB:C:D),S^{(3)}(AC:B:D),S^{(3)}(AD:B:C), \notag\\
&S^{(3)}(A:BC:D),S^{(3)}(A:BD:C),S^{(3)}(A:B:CD); \notag\\
&S^{(4)}(A:B:C:D)\big).
\label{eq:hmec4-coordinate-order}
\end{align}
The cone formed by the above $7$ orbits has \(49\) extreme
rays, organized into nine orbits. Note that here an extreme ray orbit also refers to the collection of an extreme ray and all of its permutations under $S_4$ symmetry. For each of the extreme ray orbits, we have found a corresponding graph model realization. Therefore, the \(7\) orbits above account for all facets of \(\HMEC_4\).  The exact extreme ray representatives and graph realizations are shown at orbit level in
Table~\ref{tab:hmec4-ray-reps} and
Figure~\ref{fig:hmec4-ray-graphs} in the next subsection.

Following the $n=3$ analysis, \(\HMEC_4\) should then be compared with the
standard HEC with three boundary regions and a purifier. In this case, the only new facet of HEC other than subadditivity is MMI
\cite{Hayden:2011ag}. Among the above new inequalities with four-partite multi-entropies, \(F_{4.1}\) and \(F_{4.3}\) were
independently found in recent work~\cite{BalasubramanianEtAlFourParty}; they
give a direct HMEC refinement of the standard MMI inequality
\begin{equation}
\mathrm{MMI}=
S_{AB}
+
S_{AC}
+
S_{AD}
-
S_A
-
S_B
-
S_C
-
S_D
=
-I_3(A:B:C)
\ge 0 ,
\end{equation}
obtained as
\begin{equation}
F_{4.1}
+
F_{4.3}
=
\mathrm{MMI} .
\end{equation}
Thus, the standard MMI is a positive conic combination of HMEC facet inequalities; after fixing a common normalization, it becomes a convex combination,
exactly as subadditivity is resolved in \(\HMEC_3\).

The complete \(\HMEC_3\) and \(\HMEC_4\) data are summarized in
Table~\ref{tab:low-dimensional-hmec-data}. 
\begin{table}[H]
\centering
\resizebox{\textwidth}{!}{%
\begin{tabular}{c|c|c|c|c|c|c|c}
\hline
Cone
&
Dimension
&
New facet orbits
&
Total facet orbits
&
Facets
&
Rays
&
Ray orbits
&
Graph rays
\\
\hline
\(\HMEC_3\)
&
\(4\)
&
\(2\)
&
\(2\)
&
\(4\)
&
\(4\)
&
\(2\)
&
\(2/2\)
\\
\(\HMEC_4\)
&
\(14\)
&
\(5\) basic
&
\(7=2+5\)
&
\(27\)
&
\(49\)
&
\(9\)
&
\(9/9\)
\\
\hline
\end{tabular}
}
\caption{
Low-dimensional HMEC data. The dimension refers to that of the multi-entropy vector space.  
}
\label{tab:low-dimensional-hmec-data}
\end{table}

\subsection{Extreme ray orbits and graph realization of \(\HMEC_4\)}

The seven facet-orbit representatives in the previous subsection, together with their permutations, give \(27\) concrete facet inequalities. The cone cut out by these \(27\)
inequalities has \(49\) extreme rays, grouped
into nine orbits. In Table~\ref{tab:hmec4-ray-reps}, We show one representative of each extreme ray orbit
in the coordinate order \eqref{eq:hmec4-coordinate-order}.

\begin{table}[h]
\centering
\scriptsize
\renewcommand{\arraystretch}{1.15}
\begin{tabular}{c@{\qquad}l}
\hline
orbit & representatives in the coordinate order \eqref{eq:hmec4-coordinate-order} \\
\hline
\(R_1\) & \((1,1,0,0;\;0,1,1;\;0,1,1,1,1,1;\;1)\) \\
\(R_2\) & \((1,0,1,1;\;1,1,1;\;2,1,1,2,2,1;\;2)\) \\
\(R_3\) & \((1,1,1,1;\;2,2,2;\;2,2,2,2,2,2;\;3)\) \\
\(R_4\) & \((2,1,1,1;\;2,2,2;\;2,2,2,3,3,3;\;3)\) \\
\(R_5\) & \((1,2,1,2;\;3,2,3;\;3,4,3,3,2,3;\;4)\) \\
\(R_6\) & \((1,3,2,2;\;4,3,3;\;4,5,5,3,3,4;\;5)\) \\
\(R_7\) & \((1,1,1,1;\;2,1,2;\;2,2,2,2,2,2;\;3)\) \\
\(R_8\) & \((2,1,1,2;\;3,2,2;\;3,3,2,4,3,3;\;4)\) \\
\(R_9\) & \((2,2,2,2;\;3,3,3;\;4,4,4,4,4,4;\;6)\) \\
\hline
\end{tabular}
\caption{Representatives for the \(9\) extreme ray orbits of
\(C^{\rm HMEC}_4\). The coordinate groups in each vestor separated by ``$;$" correspond respectively to
\((S_A,S_B,S_C,S_D)\), \((S_{AB},S_{AC},S_{AD})\), the six \(S^{(3)}\)
coordinates, and \(S^{(4)}(A:B:C:D)\) in the order
\eqref{eq:hmec4-coordinate-order}.}
\label{tab:hmec4-ray-reps}
\end{table}

For each representative in Table~\ref{tab:hmec4-ray-reps},
Figure~\ref{fig:hmec4-ray-graphs} displays a weighted graph whose full
multi-entropy vector agrees with it: all fourteen coordinates in \eqref{eq:hmec4-coordinate-order} are
realized by exact multiway-cut minimizations.

\begin{figure}[h]
\centering
\begin{tikzpicture}[
  scale=0.82,
  every node/.style={transform shape},
  bdot/.style={circle,fill=black,draw=black,inner sep=1.25pt},
  edge/.style={line width=0.45pt},
  elabel/.style={font=\scriptsize,inner sep=1pt,fill=white,text=red!70!black},
  rlabel/.style={font=\small}
]
% R1 = old R9
\begin{scope}[xshift=-1.05cm,yshift=0cm]
  \node[rlabel] at (-0.55,1.45) {\(R_1\)};
  \node[bdot,label={[font=\small]above:\(A\)}] (r1A) at (0,0.55) {};
  \node[bdot,label={[font=\small]above:\(B\)}] (r1B) at (1.10,0.55) {};
  \node[bdot,label={[font=\small]below:\(C\)}] (r1C) at (0,-0.85) {};
  \node[bdot,label={[font=\small]below:\(D\)}] (r1D) at (1.10,-0.85) {};
  \draw[edge] (r1A)--node[elabel,pos=.50,above] {\(1\)} (r1B);
\end{scope}

% R2 = old R7
\begin{scope}[xshift=4.75cm,yshift=0cm]
  \node[rlabel] at (-1.35,1.45) {\(R_2\)};
  \node[bdot,label={[font=\small]above:\(A\)}] (r2A) at (0,1.05) {};
  \node[bdot,label={[font=\small]right:\(B\)}] (r2B) at (1.05,0) {};
  \node[bdot,label={[font=\small]below:\(C\)}] (r2C) at (0,-1.05) {};
  \node[bdot,label={[font=\small]left:\(D\)}] (r2D) at (-1.05,0) {};
  \node[bdot] (r2o) at (0,0) {};
  \draw[edge] (r2A)--node[elabel,pos=.55,right] {\(1\)} (r2o);
  \draw[edge] (r2C)--node[elabel,pos=.55,right] {\(1\)} (r2o);
  \draw[edge] (r2D)--node[elabel,pos=.55,above] {\(1\)} (r2o);
\end{scope}

% R3 = old R5
\begin{scope}[xshift=9.50cm,yshift=0cm]
  \node[rlabel] at (-1.35,1.45) {\(R_3\)};
  \node[bdot,label={[font=\small]above:\(A\)}] (r3A) at (0,1.05) {};
  \node[bdot,label={[font=\small]right:\(B\)}] (r3B) at (1.05,0) {};
  \node[bdot,label={[font=\small]below:\(C\)}] (r3C) at (0,-1.05) {};
  \node[bdot,label={[font=\small]left:\(D\)}] (r3D) at (-1.05,0) {};
  \node[bdot] (r3o) at (0,0) {};
  \draw[edge] (r3A)--node[elabel,pos=.55,right] {\(1\)} (r3o);
  \draw[edge] (r3B)--node[elabel,pos=.55,above] {\(1\)} (r3o);
  \draw[edge] (r3C)--node[elabel,pos=.55,right] {\(1\)} (r3o);
  \draw[edge] (r3D)--node[elabel,pos=.55,above] {\(1\)} (r3o);
\end{scope}

% R4 = old R8
\begin{scope}[xshift=0cm,yshift=-3.35cm]
  \node[rlabel] at (-1.35,1.45) {\(R_4\)};
  \node[bdot,label={[font=\small]above:\(A\)}] (r4A) at (0,1.05) {};
  \node[bdot,label={[font=\small]right:\(B\)}] (r4B) at (1.05,0) {};
  \node[bdot,label={[font=\small]below:\(C\)}] (r4C) at (0,-1.05) {};
  \node[bdot,label={[font=\small]left:\(D\)}] (r4D) at (-1.05,0) {};
  \node[bdot] (r4o) at (0,0) {};
  \draw[edge] (r4A)--node[elabel,pos=.55,right] {\(2\)} (r4o);
  \draw[edge] (r4B)--node[elabel,pos=.55,above] {\(1\)} (r4o);
  \draw[edge] (r4C)--node[elabel,pos=.55,right] {\(1\)} (r4o);
  \draw[edge] (r4D)--node[elabel,pos=.55,above] {\(1\)} (r4o);
\end{scope}

% R5 = old R4
\begin{scope}[xshift=4.75cm,yshift=-3.35cm]
  \node[rlabel] at (-1.35,1.45) {\(R_5\)};
  \node[bdot,label={[font=\small]above:\(A\)}] (r5A) at (0,1.05) {};
  \node[bdot,label={[font=\small]right:\(B\)}] (r5B) at (1.05,0) {};
  \node[bdot,label={[font=\small]below:\(C\)}] (r5C) at (0,-1.05) {};
  \node[bdot,label={[font=\small]left:\(D\)}] (r5D) at (-1.05,0) {};
  \node[bdot] (r5o) at (0,0) {};
  \draw[edge] (r5A)--node[elabel,pos=.55,right] {\(1\)} (r5o);
  \draw[edge] (r5B)--node[elabel,pos=.55,above] {\(2\)} (r5o);
  \draw[edge] (r5C)--node[elabel,pos=.55,right] {\(1\)} (r5o);
  \draw[edge] (r5D)--node[elabel,pos=.55,above] {\(2\)} (r5o);
\end{scope}

% R6 = old R1
\begin{scope}[xshift=9.50cm,yshift=-3.35cm]
  \node[rlabel] at (-1.35,1.45) {\(R_6\)};
  \node[bdot,label={[font=\small]above:\(A\)}] (r6A) at (0,1.05) {};
  \node[bdot,label={[font=\small]right:\(B\)}] (r6B) at (1.05,0) {};
  \node[bdot,label={[font=\small]below:\(C\)}] (r6C) at (0,-1.05) {};
  \node[bdot,label={[font=\small]left:\(D\)}] (r6D) at (-1.05,0) {};
  \node[bdot] (r6o) at (0,0) {};
  \draw[edge] (r6A)--node[elabel,pos=.55,right] {\(1\)} (r6o);
  \draw[edge] (r6B)--node[elabel,pos=.55,above] {\(3\)} (r6o);
  \draw[edge] (r6C)--node[elabel,pos=.55,right] {\(2\)} (r6o);
  \draw[edge] (r6D)--node[elabel,pos=.55,above] {\(2\)} (r6o);
\end{scope}

% R7 = old R3
\begin{scope}[xshift=-0.65cm,yshift=-7.05cm]
  \node[rlabel] at (0.65,1.62) {\(R_7\)};
  \node[bdot,label={[font=\small]above:\(A\)}] (r7A) at (-0.55,0.95) {};
  \node[bdot,label={[font=\small]above:\(B\)}] (r7B) at (1.85,0.95) {};
  \node[bdot,label={[font=\small]below:\(C\)}] (r7C) at (-0.55,-0.95) {};
  \node[bdot,label={[font=\small]below:\(D\)}] (r7D) at (1.85,-0.95) {};
  \node[bdot] (r7p) at (0,0) {};
  \node[bdot] (r7q) at (1.30,0) {};
  \draw[edge] (r7A)--node[elabel,pos=.55,left] {\(1\)} (r7p);
  \draw[edge] (r7C)--node[elabel,pos=.55,left] {\(1\)} (r7p);
  \draw[edge] (r7B)--node[elabel,pos=.55,right] {\(1\)} (r7q);
  \draw[edge] (r7D)--node[elabel,pos=.55,right] {\(1\)} (r7q);
  \draw[edge] (r7p)--node[elabel,pos=.50,above] {\(1\)} (r7q);
\end{scope}

% R8 = old R2
\begin{scope}[xshift=4.10cm,yshift=-7.05cm]
  \node[rlabel] at (0.65,1.62) {\(R_8\)};
  \node[bdot,label={[font=\small]above:\(A\)}] (r8A) at (1.85,0.95) {};
  \node[bdot,label={[font=\small]above:\(B\)}] (r8B) at (-0.55,0.95) {};
  \node[bdot,label={[font=\small]below:\(C\)}] (r8C) at (1.85,-0.95) {};
  \node[bdot,label={[font=\small]below:\(D\)}] (r8D) at (-0.55,-0.95) {};
  \node[bdot] (r8p) at (0,0) {};
  \node[bdot] (r8q) at (1.30,0) {};
  \draw[edge] (r8B)--node[elabel,pos=.55,left] {\(1\)} (r8p);
  \draw[edge] (r8D)--node[elabel,pos=.55,left] {\(2\)} (r8p);
  \draw[edge] (r8A)--node[elabel,pos=.55,right] {\(2\)} (r8q);
  \draw[edge] (r8C)--node[elabel,pos=.55,right] {\(1\)} (r8q);
  \draw[edge] (r8p)--node[elabel,pos=.50,above] {\(2\)} (r8q);
\end{scope}

% R9 = old R6
\begin{scope}[xshift=9.50cm,yshift=-7.40cm]
  \node[rlabel] at (-1.55,2.05) {\(R_9\)};
  \node[bdot,label={[font=\small]left:\(A\)}] (r9A) at (-1.95,-0.25) {};
  \node[bdot,label={[font=\small]above:\(B\)}] (r9B) at (0,1.90) {};
  \node[bdot,label={[font=\small]right:\(C\)}] (r9C) at (2.20,0.55) {};
  \node[bdot,label={[font=\small]right:\(D\)}] (r9D) at (1.95,-0.25) {};
  \node[bdot] (r9p) at (-0.85,0.65) {};
  \node[bdot] (r9q) at (0.85,0.65) {};
  \node[bdot] (r9r) at (0,-0.85) {};
  \node[bdot] (r9s) at (0,0.10) {};
  \draw[edge] (r9A)--node[elabel,pos=.55,above left] {\(1\)} (r9p);
  \draw[edge] (r9A)--node[elabel,pos=.55,below left] {\(1\)} (r9r);
  \draw[edge] (r9B)--node[elabel,pos=.52,left] {\(1\)} (r9p);
  \draw[edge] (r9B)--node[elabel,pos=.52,right] {\(1\)} (r9q);
  \draw[edge] (r9C)--node[elabel,pos=.53,below,yshift=-1pt] {\(2\)} (r9s);
  \draw[edge] (r9D)--node[elabel,pos=.55,right] {\(1\)} (r9q);
  \draw[edge] (r9D)--node[elabel,pos=.55,left] {\(1\)} (r9r);
  \draw[edge] (r9p)--node[elabel,pos=.55,above] {\(1\)} (r9s);
  \draw[edge] (r9q)--node[elabel,pos=.55,above] {\(1\)} (r9s);
  \draw[edge] (r9r)--node[elabel,pos=.50,right] {\(1\)} (r9s);
\end{scope}
\end{tikzpicture}
\caption{Graph representatives for the \(9\) extreme ray orbits of
\(C^{\rm HMEC}_4\). The dots labeled by
\(A,B,C,D\) are the boundary vertices, while unlabeled dots represent internal graph vertices. Isolated boundary vertices are also included explicitly. Edge weights are labeled in red on the corresponding edge.  The full multi-entropy vector of each graph is the corresponding ray
listed in Table~\ref{tab:hmec4-ray-reps}.}
\label{fig:hmec4-ray-graphs}
\end{figure}
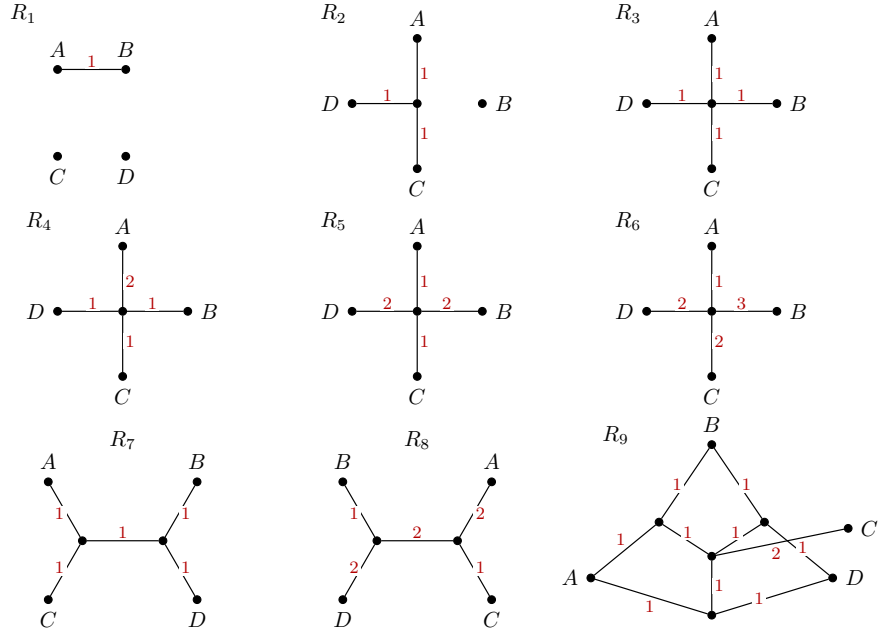
Thus, the \(27\) certified inequalities yield the
complete facets of the $n=4$ holographic multi-entropy cone.

\subsection*{Complexity of the \texorpdfstring{\(\HMEC_5\)}{HMEC5} problem}

Having completely identified the \(\HMEC_3\) and \(\HMEC_4\), the next natural step would be to study
\(\HMEC_5\). However, the calculation of this case is already comparable in size to the largest standard
HEC computations currently available: 
the holographic multi-entropy vector space where \(\HMEC_5\) lives in has dimension
\[
\dim M_5=B_5-1=51 .
\]
Its superbalanced sector already has dimension \(41\)~\cite{GeneralizedBalance2602}.
This is only one dimension below the superbalanced sector of the standard
six-region HEC, whose dimension is \(42\). \

The comparison is
instructive: six-region HEC computations have found \(1877\) HEI orbits, and a complete facet enumeration
is not currently known
\cite{HernandezCuencaHubenyJiaHEIs,HeHubenyRotaN6}.
The \(\HMEC_5\) problem is therefore not expected to be easier. Although
the number of total labels is smaller, graph realization must match all
partition-labeled multiway-cut coordinates rather than only bipartition
min-cuts; moreover, the \(q\ge3\) coordinates involve NP-hard multiway-cut
optimization. Thus a complete enumeration of \(C^{\rm HMEC}_5\) is unlikely to
be a short extension of the \(n=4\) computation.

We therefore do not attempt a complete computation of \(C^{\rm HMEC}_5\) here.
Instead, Section~\ref{sec:conjectures} uses the complete \(n=3,4\) data as a
guide to extract two structural conjectures for HMEC facets. The second conjecture concerns
the balance structure of HMEC facets and examines an \(11\)-dimensional
highly balanced subspace of its multi-entropy space, thereby revealing a concrete structural feature of \(C^{\rm HMEC}_5\).

\section{Two guiding conjectures for HMEC facets}
\label{sec:conjectures}

Having obtained the complete facets of the holographic multi-entropy cone for
\(n=3\) and \(n=4\), we now turn to the structural feature suggested by these
results. The purpose of this section is to isolate patterns
in the behavior of HMEIs and to explain how the
HMEC refines the standard holographic entropy cone. The first conjecture states
that, after a common normalization, standard HEC facet inequalities arise from
convex combinations of HMEC facet inequalities. The second is the
balanced-but-not-too-balanced conjecture: in the multi-entropy scenario, we
generalize the definition of superbalance to \(k\)-party balance and conjecture
that holographically sign-definite multipartite signals should occupy a
restricted range of party-balance sectors, rather than living in arbitrarily
high-party balanced spaces. Both conjectures are supported by the complete HMEC
results for \(n=3\) and \(n=4\), together with the \(n=5\) data discussed below.

\subsection{HEC facets as convex combinations of HMEC facets}

The first guiding conjecture concerns the relation between standard HEIs and
HMEIs. Since the HMEC enlarges the coordinate system of HEC from
bipartite entropies to multi-entropies, it is natural to ask
whether familiar HEC facets remain facets of HMEC,
or instead arise from of more elementary multi-entropy constraints. We
conjecture that the latter is the correct picture: every HEC facet inequality should be
obtained from a positive combination of HMEC facet inequalities. If this is true, it provides a new motivation for studying HMEC:
from the HMEC perspective, HEC facet inequalities are not fundamental
holographic constraints; rather, their building blocks lie in the HMEC facets.

{We now formulate this conjecture more precisely in terms of the normal vectors defining the facets of HEC and HMEC.
}We define a facet normal \(h\) as the normal multi-entropy vector of a facet
passing through the origin, so that any vector \(P\) lying on this facet satisfy $P\cdot h=0$. {With the boundary divided into $n$ subregions, including the purifier,} we now embed an HEC facet normal
\(h\) into the multi-entropy vector space by assigning zero coefficients to all
multi-entropy coordinates, and denote the resulting vector by \(\tilde h\). We conjecture that \(\tilde h\)
lies in the positive span of HMEC facet normals:
\begin{equation}
  \tilde h=\sum_a \lambda_a F_a,
  \qquad
  \lambda_a>0,
\end{equation}
where the \(F_a\) are facet normals of \(C^{\mathrm{HMEC}}_n\). Equivalently, on the right-hand-side the components along multi-entropy coordinates cancel in the sum,
leaving an entropy-sector facet normal. After fixing a common normalization,
this gives the advertised convex-combination statement for the entropy-sector
projections of HMEC facet normals. The standard entropy cone should therefore be viewed as the projection of the multi-entropy cone onto entropy coordinates. In this projection, the 
multipartite coordinates are discarded, and the finer HMEC constraints
collapse into ordinary HEC facets.

This conjecture can also be phrased as a loss-of-rank statement. {This idea is illustrated schematically through a toy model in Figure~\ref{fig:hei-hmec-schematic}, where one HEC facet, represented by the toy-model analogue $x>0$, is resolved into two finer HMEC facets, modeled by $x+y>0$ and $x-y>0$ respectively, with $y$ representing a multi-entropy coordinate.}
Away from the HEC facet, an entropy vector can, in general, be lifted to multi-entropy coordinates with independent multipartite components. However, when the entropy vector is forced to lie on the HEC facet,  we expect its freedom to have independent multipartite components to disappear. {This is immediate in the toy model: imposing both HMEC facet inequalities forces any vector on the HEC facet $x=0$ to satisfy $y=0$ as well.} In other words, the compatible multi-entropy coordinates should satisfy additional linear relations, either among the multi-entropy coordinates themselves or between multi-entropy coordinates and entropy coordinates. Thus, in the enlarged HMEC space, a standard HEC facet is detected by a drop in the number of independent multi-entropy directions.

\begin{figure}[h]
\centering

\begin{tikzpicture}[
    x=1cm,
    y=1cm,
    >=stealth,
    axis/.style={black, thick},
    hei/.style={black, very thick},
    hmecline/.style={blue!60!black, thick},
    hmecregion/.style={red!18},
    heiregion/.style={gray!12},
    labelstyle/.style={font=\small}
]

\def\xmin{-1.0}
\def\xmax{5.2}
\def\ymin{-3.0}
\def\ymax{3.0}

% ------------------------------------------------------------
% Background: HEI half-plane x>0
% ------------------------------------------------------------
\fill[heiregion]
(0,\ymin)
rectangle
(\xmax,\ymax);

\node[labelstyle, gray!60!black]
at (4.25,2.55)
{\(\mathrm{HEI}\ge 0\)};

% ------------------------------------------------------------
% HMEC region: x+y>=0 and x-y>=0, i.e. x >= |y|
% ------------------------------------------------------------
\fill[hmecregion]
(0,0)
--
(\ymax,\ymax)
--
(\xmax,\ymax)
--
(\xmax,\ymin)
--
({-\ymin},\ymin)
--
cycle;

% Optional red boundary emphasis for the filled cone
\draw[red!45!black, line width=0.6pt]
(0,0)
--
(\ymax,\ymax);

\draw[red!45!black, line width=0.6pt]
(0,0)
--
({-\ymin},\ymin);

\node[labelstyle, red!65!black]
at (3.85,0.55)
{\(\mathrm{HMEC}\)};

% ------------------------------------------------------------
% Coordinate axes
% ------------------------------------------------------------

% Negative x-axis, thin
\draw[axis]
(\xmin,0)
--
(0,0);

% Positive x-axis, thick and with arrow
\draw[hei,->]
(0,0)
--
(\xmax+0.35,0)
node[right,labelstyle]
{\(\mathrm{HEI}\)};

% y-axis, dash-dot
\draw[dash dot, thick]
(0,\ymin-0.2)
--
(0,\ymax+0.25)
node[above,labelstyle]
{multi-entropy axis};

% ------------------------------------------------------------
% Boundary lines x+y=0 and x-y=0
% ------------------------------------------------------------
\draw[hmecline]
(\xmin,-\xmin)
--
({-\ymin},\ymin)
node[above right,labelstyle]
{\(x+y=0\)};

\draw[hmecline]
(\xmin,\xmin)
--
(\ymax,\ymax)
node[below right,labelstyle]
{\(x-y=0\)};

% ------------------------------------------------------------
% Origin and optional axis marks
% ------------------------------------------------------------
\fill[black]
(0,0)
circle
(1.2pt);

\node[labelstyle, below left]
at (0,0)
{\(0\)};

\end{tikzpicture}

\caption{
{A toy model showing the} relation between an entropy-sector inequality and its multi-entropy refinement. The horizontal axis denotes the standard HEI direction, while the vertical axis represents a multi-entropy component. The red cone is cut out by the two HMEC inequalities \(x+y\ge0\) and \(x-y\ge0\), and its projection gives the standard HEI constraint \(x\ge0\).
}
\label{fig:hei-hmec-schematic}

\end{figure}
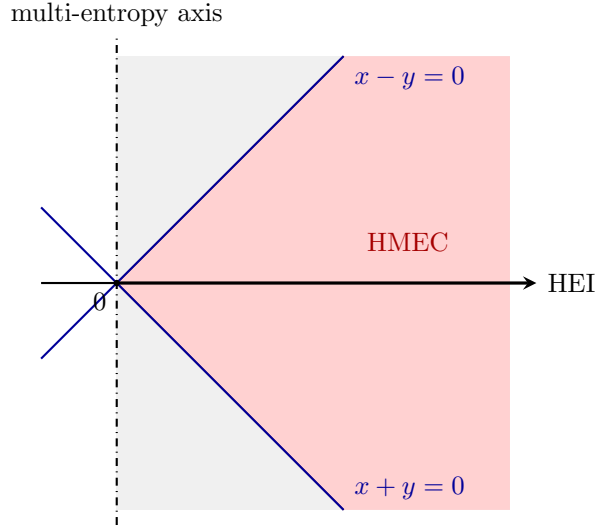

The simplest intuitive example comes from mutual information. Consider when the boundary is divided into three
boundary regions \(A,B,C\), and denote their individual entanglement wedges by
\(\operatorname{EW}(A)\), \(\operatorname{EW}(B)\), and
\(\operatorname{EW}(C)\). Heuristically, the geometric dual of
\(S^{(3)}(A:B:C)\) probes the bulk support of genuinely tripartite correlations
outside the three individual entanglement wedges,
\(\operatorname{EW}(A)\cup \operatorname{EW}(B)\cup \operatorname{EW}(C)\).
It therefore captures information that is not visible from the separate
bipartition entropy data of the individual regions, and remain invariant
under independent unitary transformations within \(A\), \(B\), and \(C\).

When an HEC facet inequality is saturated, e.g. \(S_A+S_B=S_{AB}\), the bulk
geometry undergoes a corresponding wedge factorization~\cite{CzechShuaiWangQEC,CzechShuaiRGHEC}:
schematically,
\begin{equation}
  \operatorname{EW}(A)\cup\operatorname{EW}(B)=\operatorname{EW}(AB).
\end{equation}
In this limit, the tripartite web dual to \(S^{(3)}(A:B:C)\) has less geometric room to move: instead of
residing in a generic codimension-zero bulk chamber, it is forced to follow the
factorized wedge structure, so the available room degenerates to a
lower-dimensional region determined by the interfaces of the entanglement wedges.
Thus, the saturation of the entropy facet inequality is reflected in a loss of generic
bulk room for the tripartite minimal multi-entropy surface. In such a degenerate
configuration, \(S^{(3)}(A:B:C)\) is no longer an independent
multi-entropy coordinate: its remaining contribution may be expressible in
terms of entropy data. In the present case, we have $S^{(3)}(A:B:C)=S_A+S_B,$
which gives a concrete linear relation between tripartite multi-entropy and
standard entropy coordinates on the HEC facet \(I(A:B)=0\).

Compatible evidence for the \(n=4\) HEC facet \(\MMI\) comes from the recent
four-party analysis \cite{BalasubramanianEtAlFourParty}, where known four-party entanglement signals vanish when the tripartite
information \(I_3\) vanishes. Thus, saturation of the holographic entropy
constraint \(I_3=0\) imposes additional relations among the corresponding
multi-entropy signals, providing a higher-party version of the dependent tripartite
multi-entropy coordinate picture above. This is consistent with our result in
Section~\ref{sec:low-dimensional}, where MMI is
obtained as the sum of two HMEC facets.

Therefore, the known low-party HEC facets provide the first checks of the
conjecture. The mutual-information and $I_3$ examples show how saturation of HEC facets force multi-entropy coordinates to lose their independent support and reduce to lower-order entropy data. Meanwhile, our $n=3,4$ HMEC computation show that subadditivity and MMI are both positive sums of HMEC facets. 
A systematic study of this conjecture, including higher-party HEC facets and
possible HMEC refinements, will appear in future work~\cite{CzechFengJuZhaoInProgress}.

\subsection{Balanced-but-not-too-balanced conjecture}

The second guiding conjecture starts from the well-known fact that, apart from
subadditivity, all known HEC facets are superbalanced~\cite{HeHubenyRangamaniSuperbalance}. For a standard linear entropy combination
\(Q=\sum_{\emptyset\neq I\subseteq[n]} c_I S_I\), the usual balance
condition requires \(\sum_{I\ni a} c_I=0\) for every party \(a\); the
superbalance condition requires \(Q\) to remain balanced after any exchange of a
boundary subsystem with the purifier. For example, mutual information
\begin{equation}
I_2(A:B)=S_A+S_B-S_{AB}
\end{equation}
is balanced but not superbalanced. To generalize the definition of superbalance
to the multi-entropy combination scenario, we use another property of
superbalanced combinations: when added a Bell pair between any two boundary
subsystems, although the individual entropy or multi-entropy coordinates may
change, the value of a superbalanced linear combination remains unchanged. This
zero-response language can be naturally generalized into the definition of
higher-party balanceness: a \(3\)-party balanced quantity is required to have zero
response when an arbitrary pure tripartite state, such as a GHZ state or a
\(W\) state, is added to any chosen triple of boundary subsystems. More
generally, a \(k\)-party balanced quantity has zero response to inserting an
arbitrary pure \(k\)-party auxiliary state on any chosen \(k\)-tuple of boundary
parties.

In \cite{GeneralizedBalance2602}, a framework was developed to derive all
\(k\)-party balanced multi-entropy combinations, with genuine multi-entropy
signals \cite{IizukaGenuineMultiEntropy,GaddeGenuineSignals} as special cases.
The result shows that the \(k\)-party balanced signals form a subspace in the multi-entropy signal space\footnote{{Note that this space is generally identical to the multi-entropy vector space. It contains the signals written as linear combinations of multi-entropies. The dimensions are again the nontrivial partitions corresponding to the multi-entropies, but the vectors encode the coefficients in each multi-entropy term instead of multi-entropy values, and can have negative coordinates. When there are $n$ boundary parties, we denote this signal space by MEMS$_n$ \cite{GeneralizedBalance2602}.}}. By this definition, the \(n\)-partite information
quantities \(I_3\) and \(I_4\) are \(3\)-party balanced, \(I_5\) and \(I_6\)
are \(5\)-party balanced, and the hierarchy continues to higher-party balance
levels for higher alternating information quantities. Table~\ref{tab:balance-classes} summarizes this balance hierarchy. The 3-party balanced sector has \(I_3\) and the
HMEC facets \(F_{4.1}\) and \(F_{4.3}\) as representatives. The 4-party balanced row contains the eleven-dimensional
five-label signal space analyzed below, while the final row records the formula that calculates the dimension of the $k$-party balanced sector in the multi-entropy signal space MEMS\(_n\).

In the information-basis organization of known HEC facets
~\cite{HubenyRangamaniRotaEntropyArrangement,HeHubenyRangamaniSuperbalance},
the leading structures are built from tripartite information \(I_3\). No known
facet is controlled by \(I_5\) or by a higher-party balanced information
quantity. This motivates the slogan
\emph{balanced-but-not-too-balanced}: holographically sign-definite quantities
require enough cancellation power to remove trivial lower-order responses, but
too much cancellation power can poses additional challenges for
sign-definiteness.

\begin{table}[h]
\centering
\begingroup
\small
\renewcommand{\arraystretch}{1.18}
\begin{tabular}{c|c}
\hline
\textbf{balance class} & \textbf{representative signals} \\
\hline
not superbalanced
&
\begin{tabular}{c}
\(I_2(A:B)=S_A+S_B-S_{AB}\) \\
\(S_A+S_B-S^{(3)}(A:B:C)\)
\end{tabular}
\\
\hline
2-party balance
&
\begin{tabular}{c}
\(S^{(3)}(A:B:C)-\frac12\bigl(S_A+S_B+S_C\bigr)\)\\
\(F_{4.2}\), \quad \(F_{4.4}\), \quad \(F_{4.5}\), \quad \(GM^{(3)}(A:B:C)\)
\end{tabular}
\\
\hline
3-party balance
&
\begin{tabular}{c}
\(I_3\), \quad \(F_{4.1}\), \quad \(F_{4.3}\), \quad \(GM^{(4)}(A:B:C:D)\) \\
Four-dimensional subspace in MEMS\(_4\)
\end{tabular}
\\
\hline
4-party balance
&
\begin{tabular}{c}
\(\operatorname{span}\{M_5,\mathcal E_{AB},\mathcal E_{AC},\mathcal E_{AD},\mathcal E_{AE},\) \\
\(\quad \mathcal E_{BC},\mathcal E_{BD},\mathcal E_{BE},\mathcal E_{CD},\mathcal E_{CE},\mathcal E_{DE}\}\), \quad e.g. \(GM^{(5)}\)
\end{tabular}
\\
\hline
k-party balance
&
\begin{tabular}{c}
\(\displaystyle
\binom{n}{k+1}\sum_{j=0}^{k+1}(-1)^j\binom{k+1}{j}\,B_{k-j+1}
\) \\
dimensional subspace in MEMS\(_n\)~\cite{GeneralizedBalance2602}, \quad e.g. \(GM^{(k+1)}\)
\end{tabular}
\\
\hline
\end{tabular}
\endgroup
\caption{
Reference balance classes for the balanced-but-not-too-balanced conjecture. The table separates standard balanced quantities, such as \(I_2\), from the higher-party balance hierarchy relevant for genuinely multipartite multi-entropy signals.
}
\label{tab:balance-classes}
\end{table}

\vspace{0.3cm}
\noindent\textbf{Conjecture (balanced-but-not-too-balanced).}
Apart from the $P_A$-type facets:
\begin{equation}
    S_A+S_B-S^{(3)}(A:B:C)\ge 0,
\end{equation}
together with its permutations and their lifts, every HMEC facet inequality is either
\(2\)-party balanced or \(3\)-party balanced. 
\vspace{0.3cm}

The point of the conjecture is that the allowed balance range is sharply restricted. The exceptional
$P_A$ orbit is not superbalanced; it can be viewed as the multi-entropy counterpart of subadditivity. Once this exceptional orbit is removed, the complete
\(\HMEC_3\) and \(\HMEC_4\) data show that all HMEC facets fall into the
\(2\)- or \(3\)-party balanced sectors as listed in
Table~\ref{tab:balance-classes}. In view of the fact that HEC
facets are superbalanced, apart from subadditivity
\cite{HeHubenyRangamaniSuperbalance}, the superbalance property passing from HEC to
HMEC is not by itself surprising. The genuinely nontrivial part is the phrase
\emph{not too balanced}. It says that high-party balanced signal
spaces, although well motivated by symmetry and cancellation of lower-party entanglement contributions, can be too
constrained to contain any nonzero holographically sign-definite direction. The
$\HMEC_5$ test below verifies this mechanism in the first relevant
high-party balanced sector: we prove that the \(11\)-dimensional four-party balanced subspace of
MEMS\(_5\) contains no HMEI. This provides direct
evidence that HMEC facets are confined to the \(2\)- and \(3\)-party
balanced sectors, apart from $P_A$.

Before proving the absence of HMEI in the \(11\)-dimensional four-party balanced subspace of
MEMS\(_5\), we first use a simple three-dimensional HEC as toy model to explain what happens when every linear combination inside a candidate (multi)entropy subspace fails to define an HMEC-valid inequality.
Consider the entropy cone of a pure three-party system, written in the
independent coordinates \((S_A,S_B,S_C)\), defined by
\begin{equation}
  S_A+S_B\ge S_C,\qquad
  S_A+S_C\ge S_B,\qquad
  S_B+S_C\ge S_A.
\end{equation}
Now consider the entanglement signal subspace in which a vector corresponds to an entanglement signal, as shown in Figure \ref{fig:balance-obstruction-3d}:\footnote{Strictly speaking, the entropy signal space where $M$ resides in is physically different from the entropy space that contain entropy vectors. In the former case coordinates represent coefficients, rather than entropy values. However, these spaces are formally identical to each other, and for convenience we draw $M$ and the entropy cone in the same space.}
\begin{equation}
  M=\{u_A S_A+u_B S_B+u_C S_C:\ u_A+u_B+u_C=0\},
\end{equation}
or equivalently
\begin{equation}
  M=\{u\in\mathbb R^3:\ u_A+u_B+u_C=0\}.
\end{equation}
Its normal vector is
\begin{equation}
  M^\perp=\spanop\{(1,1,1)\}.
\end{equation}
This is the entropy vector that vanishes on all signals inside $M$. It lies strictly inside the entropy cone, since \(1+1>1\)
for all three triangle inequalities. Therefore, for every entanglement signal \(u\in M\),
the hyperplane $P_u$ defined by \(u\cdot P_u=0\), that contain all entropy vectors vanishing on $u$, passes through an interior direction of the cone.
Such a hyperplane cannot support the cone; it must slice through it. Thus, some entropy vectors in the HEC give positive values of $u$, while others on the other side of $P_u$ contribute to negative $u$.
Consequently, no signal in \(M\) has a definite holographic sign. This
toy example captures the obstruction for high-party balanced subspaces of the HMEC, in alignment with our conjecture. Note that \(M\) is only an illustrative model, not itself a high-party balanced sector. It provides an analogue to the obstruction for the HMEC high-party balanced subspace:
Every nonzero signal in a
high-party balanced subspace defines a  hyperplane perpendicular to it that contains the multi-entropy vectors vanishing on this signal. However, none of these hyperplanes supports the HMEC. Instead, each one cuts through the cone, so the corresponding linear
signal has both positive and negative values on holographic
multi-entropy vectors.

\begin{figure}[h]
\centering
\tdplotsetmaincoords{72}{118}

\scalebox{0.9}{%
\begin{tikzpicture}[
    tdplot_main_coords,
    scale=1.72,
    >=stealth,
    line join=round,
    line cap=round,
    every node/.style={font=\scriptsize}
]

% ------------------------------------------------------------
% Coordinates
% ------------------------------------------------------------

\coordinate (O) at (0,0,0);

% Axes
\coordinate (XA) at (3.4,0,0);
\coordinate (XB) at (0,3.4,0);
\coordinate (XC) at (0,0,3.4);

% HEC extreme rays
\coordinate (RAB) at (2.7,2.7,0);
\coordinate (RAC) at (2.7,0,2.7);
\coordinate (RBC) at (0,2.7,2.7);

% Plane M : SA+SB+SC = 0
\coordinate (M1) at ( 2.7,-0.8,-1.9);
\coordinate (M2) at (-0.8, 2.7,-1.9);
\coordinate (M3) at (-2.7, 0.8, 1.9);
\coordinate (M4) at ( 0.8,-2.7, 1.9);

% Interior normal n=(1,1,1)
\coordinate (N) at (1.55,1.55,1.55);

% Choose u in M
\coordinate (U) at (1.75,-1.75,0);

% Plane Pu : u.S = 0, here u=(1,-1,0), so x-y=0
\coordinate (P1) at ( 2.30, 2.30,-1.20);
\coordinate (P2) at ( 0.00, 0.00, 3.20);
\coordinate (P3) at (-2.30,-2.30, 1.20);
\coordinate (P4) at ( 0.00, 0.00,-3.20);

% Three sample points on cone
\coordinate (C1) at (2.7,2.7,0);
\coordinate (C2) at (2.7,0,2.7);
\coordinate (C3) at (0,2.7,2.7);

% Orthogonal projections to M along (1,1,1)
\coordinate (Q1) at ( 0.9, 0.9,-1.8);
\coordinate (Q2) at ( 0.9,-1.8, 0.9);
\coordinate (Q3) at (-1.8, 0.9, 0.9);

% Exact intersection of Pu with top cap of the truncated cone
% CutA is RAB, and CutB is midpoint of edge RAC--RBC
\coordinate (CutA) at (2.7,2.7,0);
\coordinate (CutB) at (1.35,1.35,2.7);

% ------------------------------------------------------------
% Background planes
% ------------------------------------------------------------

% Plane M
\filldraw[
    fill=green!18,
    draw=green!45!black,
    opacity=0.38
]
(M1)--(M2)--(M3)--(M4)--cycle;

% ------------------------------------------------------------
% HEC cone (truncated)
% ------------------------------------------------------------

\filldraw[
    fill=red!22,
    draw=red!60!black,
    opacity=0.34
]
(O)--(RAB)--(RAC)--cycle;

\filldraw[
    fill=red!22,
    draw=red!60!black,
    opacity=0.34
]
(O)--(RAC)--(RBC)--cycle;

\filldraw[
    fill=red!22,
    draw=red!60!black,
    opacity=0.34
]
(O)--(RBC)--(RAB)--cycle;

\filldraw[
    fill=red!14,
    draw=red!60!black,
    opacity=0.28
]
(RAB)--(RAC)--(RBC)--cycle;

% ------------------------------------------------------------
% Slice Pu ∩ HEC : draw it as a face lying in Pu
% ------------------------------------------------------------

\filldraw[
    fill=blue!45,
    draw=black,
    opacity=0.32,
    dash dot,
    line width=0.9pt
]
(O)--(CutA)--(CutB)--cycle;

% emphasize boundary of the slice
\draw[dash dot, very thick, black] (O)--(CutA);
\draw[dash dot, very thick, black] (O)--(CutB);
\draw[dash dot, very thick, black] (CutA)--(CutB);

% ------------------------------------------------------------
% Axes
% ------------------------------------------------------------

\draw[->, thick] (O)--(XA) node[below left] {$S_A$};
\draw[->, thick] (O)--(XB) node[below right] {$S_B$};
\draw[->, thick] (O)--(XC) node[above] {$S_C=S_{AB}$};

% ------------------------------------------------------------
% Labels
% ------------------------------------------------------------

\node[green!50!black] at (1.25,-0.75,-1.35)
{$M:\ u_A+u_B+u_C=0$};

\node[black] at (2.72,2.58,2.80)
{$P_u\cap \mathrm{HEC}$};

% ------------------------------------------------------------
% extreme rays
% ------------------------------------------------------------

% ------------------------------------------------------------
% Normal vector to M
% ------------------------------------------------------------

\draw[->, very thick, purple!80!black]
(O)--(N);

\node[purple!70!black] at (0.50,0.96,0.62)
{$(1,1,1)$};

% ------------------------------------------------------------
% Vector u in M
% ------------------------------------------------------------

\draw[->, very thick, orange!85!black]
(O)--(U)
node[pos=1.03, below] {$u\in M$};

% ------------------------------------------------------------
% Projection lines to M
% ------------------------------------------------------------

\draw[dashed, thick, gray!75!black] (C1)--(Q1);
\draw[dashed, thick, gray!75!black] (C2)--(Q2);
\draw[dashed, thick, gray!75!black] (C3)--(Q3);

\draw[dashed, thick, gray!80!black] (Q1)--(Q2)--(Q3)--cycle;

\fill[gray!80!black] (Q1) circle (0.45pt);
\fill[gray!80!black] (Q2) circle (0.45pt);
\fill[gray!80!black] (Q3) circle (0.45pt);

\node[gray!70!black] at (0.58,0.42,-0.82)
{projection in $M$};

% ------------------------------------------------------------
% Special points
% ------------------------------------------------------------

\fill[black] (O) circle (0.55pt);
\node[below left] at (O) {$0$};

\fill[red!70!black] (C1) circle (0.5pt);
\fill[red!70!black] (C2) circle (0.5pt);
\fill[red!70!black] (C3) circle (0.5pt);

\end{tikzpicture}
}

\caption{
Geometric illustration of the toy model depicting the obstruction mechanism in the three-party entropy cone.
The red cone is the HEC, spanned by the three extreme rays along the pairwise angle bisectors.
The green plane \(M:\,u_A+u_B+u_C=0\) is the candidate signal subspace, whose normal vector \((1,1,1)\), shown in purple, lies strictly inside the cone.
A nonzero vector \(u\in M\), shown in orange, defines the perpendicular hyperplane \(P_u:\,u\cdot P_u=0\).
Since this hyperplane contains the interior direction \((1,1,1)\), its intersection with the cone, highlighted as the dash-dotted blue wedge \(P_u\cap\mathrm{HEC}\), cuts through the cone rather than supporting it.
The gray dashed lines project three sample cone points onto \(M\), forming the dashed triangle in the candidate signal plane.
}
\label{fig:balance-obstruction-3d}
\end{figure}
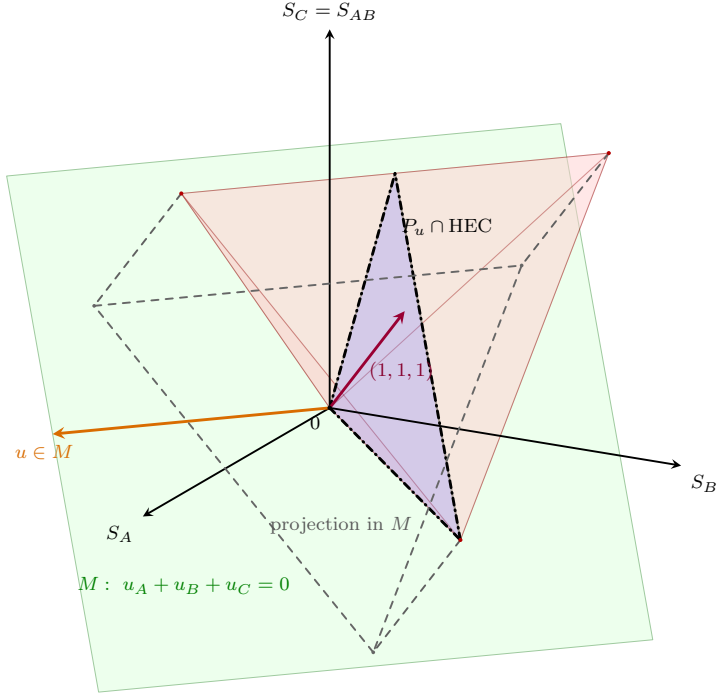

Let us generalize this toy-model argument into genuine, higher dimensional cases. Let \(C\) denote the multi-entropy cone, and let \(M\) be a linear subspace of the multi-entropy signal space. $M$ contains candidate signals that, if sign-definite, will become valid inequalities. A sufficient geometric
criterion for \(M\) to contain no such nontrivial valid inequality is:
\begin{equation}
  M^\perp\cap\operatorname{int}(C)\neq\varnothing.
  \label{eq:no-valid-normal-interior}
\end{equation}
Indeed, if the normal multi-entropy vector \(M^\perp\) intersects the interior of \(C\), then
every nonzero signal in \(M\) defines a perpendicular hyperplane passing through an interior
point of the cone. Such a hyperplane cannot be a supporting hyperplane of
\(C\); instead, it necessarily slices through the cone. Consequently,
no nonzero signal of \(M\) can have a definite holographic sign on \(C\).

{However, testing whether a normal multi-entropy vector lies strictly inside the cone is not
easy in high-dimensional cases. As a result, we use the following projected
version of the same criterion. Let
\(\pi_M:C\to M\) denote the projection onto the candidate signal subspace
\(M\). As is shown later, by choosing a basis of independent signals in $M$, for any multi-entropy vector $c\in C$, the projection outcome $\pi_M(c)$ is a vector whose coordinates are given by the values of each basis signal calculated on $c$. Thus, a practical sufficient certificate for the test is that the projected cone
\(\pi_M(C)\) contains the origin in its interior. Equivalently, if
\(d=\dim M\), it is enough to find \(d+1\) projected cone points
\(y_1,\ldots,y_{d+1}\in M\) and positive rational numbers \(q_i\) such that
\begin{equation}
  q_i>0,\qquad
  \sum_i q_i=1,\qquad
  \sum_i q_i y_i=0,\qquad
  \operatorname{rank}_{\mathrm{aff}}\{y_i\}=d.
  \label{eq:positive-simplex-certificate}
\end{equation}
Here \(\operatorname{rank}_{\mathrm{aff}}\{y_i\}\) denotes the affine rank, namely the rank of the difference vectors \(\{y_i-y_1\}_{i=2}^{d+1}\).
The points \(y_i\) then form a simplex whose interior contains the origin,
represented in Figure~\ref{fig:balance-obstruction-3d} by the dashed triangle on
the plane \(M\). Every hyperplane through the origin that is perpendicular to $M$ cuts this
simplex, and therefore cuts the projected cone rather than supporting it. Hence
every nonzero signal in \(M\) takes both positive and negative values on \(C\),
so no such signal defines a valid holographic inequality.}

We now apply the projected-simplex criterion to the first high-party balanced sector
where the obstruction can be tested explicitly: {the $4$-party balanced subspace in Table~\ref{tab:balance-classes}, with the boundary divided into five subregions.} We use the labels \(A,B,C,D,E\), where \(E\) may be regarded as the
purifier. The relevant eleven-dimensional signal space is
\begin{equation}
V_{4\text{-bal}}:
=
\operatorname{span}\{
M_5,\mathcal E_{AB},\mathcal E_{AC},\mathcal E_{AD},\mathcal E_{AE},
\mathcal E_{BC},\mathcal E_{BD},\mathcal E_{BE},\mathcal E_{CD},\mathcal E_{CE},\mathcal E_{DE}
\}.
\end{equation}
Here \(M_5\) is the fully alternating five-label signal, and
\(\mathcal E_{AB}\) is the pair-splitting signal associated with the pair \(\{A,B\}\),
with the remaining \(\mathcal E_{ij}\) obtained by permutation. Specifically, if \(\Pi_k\) denotes the set of \(k\)-block partitions of \(\{A,B,C,D,E\}\), then
\begin{equation}
M_5
=
-\sum_{\pi\in\Pi_2} S^{(2)}(\pi)
+
2\sum_{\pi\in\Pi_3} S^{(3)}(\pi)
-
6\sum_{\pi\in\Pi_4} S^{(4)}(\pi)
+
24\,S^{(5)}(A:B:C:D:E).
\end{equation}
For the representative pair \(\{A,B\}\), whose complement is \(\{C,D,E\}\), define
\begin{equation}
\Delta_{AB}(\rho)
=
S^{(|\rho|+2)}(A:B:\rho)
-
S^{(|\rho|+1)}(AB:\rho),
\end{equation}
where \(\rho\) is a partition of \(\{C,D,E\}\). The corresponding pair-splitting signal is
\begin{equation}
\mathcal E_{AB}
=
\Delta_{AB}(CDE)
-
\Delta_{AB}(C:DE)
-
\Delta_{AB}(D:CE)
-
\Delta_{AB}(E:CD)
+
2\Delta_{AB}(C:D:E).
\end{equation}
All other \(\mathcal E_{ij}\) are obtained simply by relabeling the five
boundary labels \(A,B,C,D,E\).

We now give an explicit finite certificate for the projected-simplex criterion.
For each graph \(G_i\) in Figure~\ref{fig:n5-top11-graphs},
we compute the complete five-label graph multi-entropy vector by exact multiway-cut minimization
and project it to \(V_{4\text{-bal}}\). In the ordered basis
\begin{equation}
\mathcal B_{4\text{-bal}}
=
(M_5,\mathcal E_{AB},\mathcal E_{AC},\mathcal E_{AD},\mathcal E_{AE},
\mathcal E_{BC},\mathcal E_{BD},\mathcal E_{BE},\mathcal E_{CD},\mathcal E_{CE},\mathcal E_{DE}),
\end{equation}
the projected vectors \(y_i\) are listed in Table~\ref{tab:n5-top11-simplex}.

\begin{table}[H]
\centering
\small
\setlength{\tabcolsep}{3pt}
\renewcommand{\arraystretch}{1.08}
\begin{tabular}{c c l}
\hline
graph & \(n_i\) in \(q_i=n_i/2218\) & projected vector \(y_i\) in the basis \(\mathcal B_{4\text{-bal}}\) \\
\hline
\(G_1\) & 75 & \((21,4,2,2,2,1,1,1,3,3,3)\) \\
\(G_2\) & 293 & \((-2,0,0,0,0,0,0,0,-1,0,0)\) \\
\(G_3\) & 212 & \((-2,0,-1,0,0,0,0,0,0,-1,0)\) \\
\(G_4\) & 873 & \((1,0,0,0,0,0,0,0,0,0,0)\) \\
\(G_5\) & 19 & \((21,1,1,1,4,3,3,2,3,2,2)\) \\
\(G_6\) & 266 & \((-4,0,0,-1,0,0,-1,0,0,0,-1)\) \\
\(G_7\) & 11 & \((-6,-2,1,1,0,-2,-2,1,1,0,0)\) \\
\(G_8\) & 70 & \((7,-1,1,2,0,1,2,1,0,2,0)\) \\
\(G_9\) & 191 & \((-4,0,0,0,-1,0,0,-1,0,-1,0)\) \\
\(G_{10}\) & 19 & \((-13,-2,-2,-3,-2,0,1,0,0,0,0)\) \\
\(G_{11}\) & 3 & \((0,-1,0,1,1,2,-1,-1,0,0,1)\) \\
\(G_{12}\) & 186 & \((-1,-1,0,0,0,-1,0,0,0,0,0)\) \\
\hline
\end{tabular}%
\caption{The twelve projected graph vectors used to show the five-label four-party balanced obstruction. The basis is \(\mathcal B_{4\text{-bal}}=(M_5,\mathcal E_{AB},\mathcal E_{AC},\mathcal E_{AD},\mathcal E_{AE},\mathcal E_{BC},\mathcal E_{BD},\mathcal E_{BE},\mathcal E_{CD},\mathcal E_{CE},\mathcal E_{DE})\).}
\label{tab:n5-top11-simplex}
\end{table}
\begin{figure}[!t]
\centering
\resizebox{0.72\textwidth}{!}{%
\begin{tikzpicture}[
  bdry/.style={circle, fill=black, inner sep=1.25pt},
  intv/.style={circle, fill=black, inner sep=1.35pt},
  edge/.style={line width=0.35pt},
  weight/.style={midway, text=red!70!black, inner sep=0pt, font=\tiny},
  graphlabel/.style={font=\scriptsize}
]
% Graph 1
\begin{scope}[shift={(0.00,0.00)}]
  \node[graphlabel] at (0,1.95) {$G_{1}$};
  \node[bdry] (g1A) at (162:1.55) {};
  \node[above left, font=\scriptsize, inner sep=1pt] at (g1A) {$A$};
  \node[bdry] (g1B) at (90:1.55) {};
  \node[above, font=\scriptsize, inner sep=1pt] at (g1B) {$B$};
  \node[bdry] (g1C) at (18:1.55) {};
  \node[above right, font=\scriptsize, inner sep=1pt] at (g1C) {$C$};
  \node[bdry] (g1D) at (-54:1.55) {};
  \node[below right, font=\scriptsize, inner sep=1pt] at (g1D) {$D$};
  \node[bdry] (g1E) at (-126:1.55) {};
  \node[below left, font=\scriptsize, inner sep=1pt] at (g1E) {$E$};
  \node[intv] (g1u) at (0.00,0.00) {};
  \draw[edge] (g1A) -- node[weight,pos=0.29] {1} (g1C);
  \draw[edge] (g1A) -- node[weight,pos=0.62] {5} (g1D);
  \draw[edge] (g1A) -- node[weight,pos=0.47] {2} (g1u);
  \draw[edge] (g1B) -- node[weight,pos=0.62] {2} (g1D);
  \draw[edge] (g1B) -- node[weight,pos=0.71] {1} (g1E);
  \draw[edge] (g1B) -- node[weight,pos=0.53] {5} (g1u);
  \draw[edge] (g1C) -- node[weight,pos=0.53] {4} (g1u);
  \draw[edge] (g1D) -- node[weight,pos=0.57] {4} (g1E);
  \draw[edge] (g1D) -- node[weight,pos=0.61] {4} (g1u);
  \draw[edge] (g1E) -- node[weight,pos=0.47] {4} (g1u);
\end{scope}

% Graph 2
\begin{scope}[shift={(4.00,0.00)}]
  \node[graphlabel] at (0,1.95) {$G_{2}$};
  \node[bdry] (g2A) at (162:1.55) {};
  \node[above left, font=\scriptsize, inner sep=1pt] at (g2A) {$A$};
  \node[bdry] (g2B) at (90:1.55) {};
  \node[above, font=\scriptsize, inner sep=1pt] at (g2B) {$B$};
  \node[bdry] (g2C) at (18:1.55) {};
  \node[above right, font=\scriptsize, inner sep=1pt] at (g2C) {$C$};
  \node[bdry] (g2D) at (-54:1.55) {};
  \node[below right, font=\scriptsize, inner sep=1pt] at (g2D) {$D$};
  \node[bdry] (g2E) at (-126:1.55) {};
  \node[below left, font=\scriptsize, inner sep=1pt] at (g2E) {$E$};
  \node[intv] (g2u) at (0.00,0.00) {};
  \draw[edge] (g2A) -- node[weight,pos=0.29] {2} (g2C);
  \draw[edge] (g2A) -- node[weight,pos=0.62] {3} (g2D);
  \draw[edge] (g2A) -- node[weight,pos=0.47] {1} (g2u);
  \draw[edge] (g2B) -- node[weight,pos=0.50] {3} (g2C);
  \draw[edge] (g2B) -- node[weight,pos=0.62] {3} (g2D);
  \draw[edge] (g2B) -- node[weight,pos=0.71] {5} (g2E);
  \draw[edge] (g2B) -- node[weight,pos=0.53] {1} (g2u);
  \draw[edge] (g2C) -- node[weight,pos=0.43] {3} (g2D);
  \draw[edge] (g2C) -- node[weight,pos=0.53] {3} (g2u);
  \draw[edge] (g2D) -- node[weight,pos=0.61] {5} (g2u);
  \draw[edge] (g2E) -- node[weight,pos=0.47] {4} (g2u);
\end{scope}

% Graph 3
\begin{scope}[shift={(8.00,0.00)}]
  \node[graphlabel] at (0,1.95) {$G_{3}$};
  \node[bdry] (g3A) at (162:1.55) {};
  \node[above left, font=\scriptsize, inner sep=1pt] at (g3A) {$A$};
  \node[bdry] (g3B) at (90:1.55) {};
  \node[above, font=\scriptsize, inner sep=1pt] at (g3B) {$B$};
  \node[bdry] (g3C) at (18:1.55) {};
  \node[above right, font=\scriptsize, inner sep=1pt] at (g3C) {$C$};
  \node[bdry] (g3D) at (-54:1.55) {};
  \node[below right, font=\scriptsize, inner sep=1pt] at (g3D) {$D$};
  \node[bdry] (g3E) at (-126:1.55) {};
  \node[below left, font=\scriptsize, inner sep=1pt] at (g3E) {$E$};
  \node[intv] (g3u) at (0.00,0.00) {};
  \draw[edge] (g3A) -- node[weight,pos=0.54] {3} (g3B);
  \draw[edge] (g3A) -- node[weight,pos=0.38] {1} (g3C);
  \draw[edge] (g3A) -- node[weight,pos=0.61] {1} (g3u);
  \draw[edge] (g3B) -- node[weight,pos=0.46] {2} (g3C);
  \draw[edge] (g3B) -- node[weight,pos=0.41] {2} (g3u);
  \draw[edge] (g3C) -- node[weight,pos=0.69] {4} (g3u);
  \draw[edge] (g3D) -- node[weight,pos=0.50] {1} (g3E);
  \draw[edge] (g3D) -- node[weight,pos=0.50] {2} (g3u);
  \draw[edge] (g3E) -- node[weight,pos=0.50] {1} (g3u);
\end{scope}

% Graph 4
\begin{scope}[shift={(0.00,-3.80)}]
  \node[graphlabel] at (0,1.95) {$G_{4}$};
  \node[bdry] (g4A) at (162:1.55) {};
  \node[above left, font=\scriptsize, inner sep=1pt] at (g4A) {$A$};
  \node[bdry] (g4B) at (90:1.55) {};
  \node[above, font=\scriptsize, inner sep=1pt] at (g4B) {$B$};
  \node[bdry] (g4C) at (18:1.55) {};
  \node[above right, font=\scriptsize, inner sep=1pt] at (g4C) {$C$};
  \node[bdry] (g4D) at (-54:1.55) {};
  \node[below right, font=\scriptsize, inner sep=1pt] at (g4D) {$D$};
  \node[bdry] (g4E) at (-126:1.55) {};
  \node[below left, font=\scriptsize, inner sep=1pt] at (g4E) {$E$};
  \node[intv] (g4u) at (0.00,0.00) {};
  \draw[edge] (g4A) -- node[weight,pos=0.38] {2} (g4C);
  \draw[edge] (g4A) -- node[weight,pos=0.38] {2} (g4D);
  \draw[edge] (g4A) -- node[weight,pos=0.57] {4} (g4E);
  \draw[edge] (g4A) -- node[weight,pos=0.61] {1} (g4u);
  \draw[edge] (g4B) -- node[weight,pos=0.62] {3} (g4D);
  \draw[edge] (g4B) -- node[weight,pos=0.53] {5} (g4u);
  \draw[edge] (g4C) -- node[weight,pos=0.53] {2} (g4u);
  \draw[edge] (g4D) -- node[weight,pos=0.54] {1} (g4E);
  \draw[edge] (g4D) -- node[weight,pos=0.61] {2} (g4u);
  \draw[edge] (g4E) -- node[weight,pos=0.41] {1} (g4u);
\end{scope}

% Graph 5
\begin{scope}[shift={(4.00,-3.80)}]
  \node[graphlabel] at (0,1.95) {$G_{5}$};
  \node[bdry] (g5A) at (162:1.55) {};
  \node[above left, font=\scriptsize, inner sep=1pt] at (g5A) {$A$};
  \node[bdry] (g5B) at (90:1.55) {};
  \node[above, font=\scriptsize, inner sep=1pt] at (g5B) {$B$};
  \node[bdry] (g5C) at (18:1.55) {};
  \node[above right, font=\scriptsize, inner sep=1pt] at (g5C) {$C$};
  \node[bdry] (g5D) at (-54:1.55) {};
  \node[below right, font=\scriptsize, inner sep=1pt] at (g5D) {$D$};
  \node[bdry] (g5E) at (-126:1.55) {};
  \node[below left, font=\scriptsize, inner sep=1pt] at (g5E) {$E$};
  \node[intv] (g5u) at (0.00,0.00) {};
  \draw[edge] (g5A) -- node[weight,pos=0.50] {4} (g5B);
  \draw[edge] (g5A) -- node[weight,pos=0.62] {4} (g5C);
  \draw[edge] (g5A) -- node[weight,pos=0.53] {5} (g5u);
  \draw[edge] (g5B) -- node[weight,pos=0.62] {5} (g5E);
  \draw[edge] (g5B) -- node[weight,pos=0.53] {4} (g5u);
  \draw[edge] (g5C) -- node[weight,pos=0.54] {4} (g5D);
  \draw[edge] (g5C) -- node[weight,pos=0.62] {5} (g5E);
  \draw[edge] (g5C) -- node[weight,pos=0.61] {4} (g5u);
  \draw[edge] (g5D) -- node[weight,pos=0.43] {1} (g5E);
  \draw[edge] (g5D) -- node[weight,pos=0.41] {4} (g5u);
  \draw[edge] (g5E) -- node[weight,pos=0.61] {2} (g5u);
\end{scope}

% Graph 6
\begin{scope}[shift={(8.00,-3.80)}]
  \node[graphlabel] at (0,1.95) {$G_{6}$};
  \node[bdry] (g6A) at (162:1.55) {};
  \node[above left, font=\scriptsize, inner sep=1pt] at (g6A) {$A$};
  \node[bdry] (g6B) at (90:1.55) {};
  \node[above, font=\scriptsize, inner sep=1pt] at (g6B) {$B$};
  \node[bdry] (g6C) at (18:1.55) {};
  \node[above right, font=\scriptsize, inner sep=1pt] at (g6C) {$C$};
  \node[bdry] (g6D) at (-54:1.55) {};
  \node[below right, font=\scriptsize, inner sep=1pt] at (g6D) {$D$};
  \node[bdry] (g6E) at (-126:1.55) {};
  \node[below left, font=\scriptsize, inner sep=1pt] at (g6E) {$E$};
  \node[intv] (g6u) at (0.00,0.00) {};
  \draw[edge] (g6A) -- node[weight,pos=0.50] {5} (g6B);
  \draw[edge] (g6A) -- node[weight,pos=0.29] {5} (g6C);
  \draw[edge] (g6A) -- node[weight,pos=0.47] {1} (g6u);
  \draw[edge] (g6B) -- node[weight,pos=0.29] {5} (g6D);
  \draw[edge] (g6B) -- node[weight,pos=0.62] {3} (g6E);
  \draw[edge] (g6B) -- node[weight,pos=0.47] {1} (g6u);
  \draw[edge] (g6C) -- node[weight,pos=0.62] {4} (g6E);
  \draw[edge] (g6C) -- node[weight,pos=0.53] {3} (g6u);
  \draw[edge] (g6D) -- node[weight,pos=0.43] {1} (g6E);
  \draw[edge] (g6D) -- node[weight,pos=0.53] {4} (g6u);
  \draw[edge] (g6E) -- node[weight,pos=0.61] {1} (g6u);
\end{scope}
\end{tikzpicture}%
}
\vspace{0.8em}

\resizebox{0.72\textwidth}{!}{%
\begin{tikzpicture}[
  bdry/.style={circle, fill=black, inner sep=1.25pt},
  intv/.style={circle, fill=black, inner sep=1.35pt},
  edge/.style={line width=0.35pt},
  weight/.style={midway, text=red!70!black, inner sep=0pt, font=\fontsize{5}{5}\selectfont},
  graphlabel/.style={font=\scriptsize}
]
% Graph 7
\begin{scope}[shift={(0.00,0.00)}]
  \node[graphlabel] at (0,1.95) {$G_{7}$};
  \node[bdry] (g7A) at (-1.40,1.05) {};
  \node[above left, font=\scriptsize, inner sep=1pt] at (g7A) {$A$};
  \node[bdry] (g7B) at (0.00,1.45) {};
  \node[above, font=\scriptsize, inner sep=1pt] at (g7B) {$B$};
  \node[bdry] (g7C) at (1.40,1.05) {};
  \node[above right, font=\scriptsize, inner sep=1pt] at (g7C) {$C$};
  \node[bdry] (g7D) at (0.95,-1.15) {};
  \node[below right, font=\scriptsize, inner sep=1pt] at (g7D) {$D$};
  \node[bdry] (g7E) at (-0.95,-1.15) {};
  \node[below left, font=\scriptsize, inner sep=1pt] at (g7E) {$E$};
  \node[intv] (g7u) at (-0.45,0.05) {};
  \node[intv] (g7v) at (0.48,0.10) {};
  \draw[edge] (g7A) -- node[weight,pos=0.44] {2} (g7B);
  \draw[edge] (g7A) -- node[weight,pos=0.31] {3} (g7C);
  \draw[edge] (g7A)
    .. controls (-0.83,-0.57) and (0.31,-1.00) ..
    node[weight,pos=0.46] {3} (g7D);
  \draw[edge] (g7A) -- node[weight,pos=0.54] {2} (g7E);
  \draw[edge] (g7A) -- node[weight,pos=0.47] {3} (g7u);
  \draw[edge] (g7B) to[bend right=10] node[weight,pos=0.39] {3} (g7E);
  \draw[edge] (g7B) -- node[weight,pos=0.44] {4} (g7u);
  \draw[edge] (g7B) -- node[weight,pos=0.55] {3} (g7v);
  \draw[edge] (g7C) -- node[weight,pos=0.50] {5} (g7D);
  \draw[edge] (g7C) -- node[weight,pos=0.34] {3} (g7u);
  \draw[edge] (g7C) -- node[weight,pos=0.56] {1} (g7v);
  \draw[edge] (g7D) -- node[weight,pos=0.50] {4} (g7E);
  \draw[edge] (g7D) -- node[weight,pos=0.38] {3} (g7u);
  \draw[edge] (g7E) -- node[weight,pos=0.36] {2} (g7u);
  \draw[edge] (g7E) -- node[weight,pos=0.64] {1} (g7v);
  \draw[edge] (g7u) -- node[weight,pos=0.50] {1} (g7v);
\end{scope}

% Graph 8
\begin{scope}[shift={(4.00,0.00)}]
  \node[graphlabel] at (0,1.95) {$G_{8}$};
  \node[bdry] (g8A) at (-1.40,1.05) {};
  \node[above left, font=\scriptsize, inner sep=1pt] at (g8A) {$A$};
  \node[bdry] (g8B) at (0.00,1.45) {};
  \node[above, font=\scriptsize, inner sep=1pt] at (g8B) {$B$};
  \node[bdry] (g8C) at (1.40,1.05) {};
  \node[above right, font=\scriptsize, inner sep=1pt] at (g8C) {$C$};
  \node[bdry] (g8D) at (0.95,-1.15) {};
  \node[below right, font=\scriptsize, inner sep=1pt] at (g8D) {$D$};
  \node[bdry] (g8E) at (-0.95,-1.15) {};
  \node[below left, font=\scriptsize, inner sep=1pt] at (g8E) {$E$};
  \node[intv] (g8u) at (-0.45,0.05) {};
  \node[intv] (g8v) at (0.48,0.10) {};
  \draw[edge] (g8A) -- node[weight,pos=0.50] {5} (g8u);
  \draw[edge] (g8B) -- node[weight,pos=0.56] {4} (g8C);
  \draw[edge] (g8B) to[bend left=11] node[weight,pos=0.32] {3} (g8D);
  \draw[edge] (g8B) -- node[weight,pos=0.42] {2} (g8u);
  \draw[edge] (g8B) -- node[weight,pos=0.60] {3} (g8v);
  \draw[edge] (g8C) -- node[weight,pos=0.55] {4} (g8D);
  \draw[edge] (g8C)
    .. controls (0.86,-0.59) and (-0.28,-1.01) ..
    node[weight,pos=0.44] {2} (g8E);
  \draw[edge] (g8C) -- node[weight,pos=0.45] {4} (g8v);
  \draw[edge] (g8D) -- node[weight,pos=0.50] {4} (g8E);
  \draw[edge] (g8D) -- node[weight,pos=0.43] {1} (g8u);
  \draw[edge] (g8E) -- node[weight,pos=0.43] {4} (g8u);
  \draw[edge] (g8u) -- node[weight,pos=0.50] {4} (g8v);
\end{scope}

% Graph 9
\begin{scope}[shift={(8.00,0.00)}]
  \node[graphlabel] at (0,1.95) {$G_{9}$};
  \node[bdry] (g9A) at (-1.40,1.05) {};
  \node[above left, font=\scriptsize, inner sep=1pt] at (g9A) {$A$};
  \node[bdry] (g9B) at (0.00,1.45) {};
  \node[above, font=\scriptsize, inner sep=1pt] at (g9B) {$B$};
  \node[bdry] (g9C) at (1.40,1.05) {};
  \node[above right, font=\scriptsize, inner sep=1pt] at (g9C) {$C$};
  \node[bdry] (g9D) at (0.95,-1.15) {};
  \node[below right, font=\scriptsize, inner sep=1pt] at (g9D) {$D$};
  \node[bdry] (g9E) at (-0.95,-1.15) {};
  \node[below left, font=\scriptsize, inner sep=1pt] at (g9E) {$E$};
  \node[intv] (g9u) at (-0.45,0.05) {};
  \node[intv] (g9v) at (0.48,0.10) {};
  \draw[edge] (g9A) -- node[weight,pos=0.43] {5} (g9B);
  \draw[edge] (g9A) -- node[weight,pos=0.27] {3} (g9C);
  \draw[edge] (g9A) -- node[weight,pos=0.50] {3} (g9E);
  \draw[edge] (g9A) -- node[weight,pos=0.50] {2} (g9u);
  \draw[edge] (g9B) -- node[weight,pos=0.58] {3} (g9C);
  \draw[edge] (g9B) -- node[weight,pos=0.43] {1} (g9u);
  \draw[edge] (g9C) -- node[weight,pos=0.50] {1} (g9D);
  \draw[edge] (g9C)
    .. controls (0.86,-0.59) and (-0.28,-1.01) ..
    node[weight,pos=0.44] {1} (g9E);
  \draw[edge] (g9C) -- node[weight,pos=0.37] {1} (g9u);
  \draw[edge] (g9C) -- node[weight,pos=0.49] {1} (g9v);
  \draw[edge] (g9D) -- node[weight,pos=0.49] {1} (g9E);
  \draw[edge] (g9D) -- node[weight,pos=0.39] {3} (g9u);
  \draw[edge] (g9E) -- node[weight,pos=0.43] {5} (g9u);
  \draw[edge] (g9u) -- node[weight,pos=0.50] {3} (g9v);
\end{scope}

% Graph 10
\begin{scope}[shift={(0.00,-3.80)}]
  \node[graphlabel] at (0,1.95) {$G_{10}$};
  \node[bdry] (g10A) at (-1.40,1.05) {};
  \node[above left, font=\scriptsize, inner sep=1pt] at (g10A) {$A$};
  \node[bdry] (g10B) at (0.00,1.45) {};
  \node[above, font=\scriptsize, inner sep=1pt] at (g10B) {$B$};
  \node[bdry] (g10C) at (1.40,1.05) {};
  \node[above right, font=\scriptsize, inner sep=1pt] at (g10C) {$C$};
  \node[bdry] (g10D) at (0.95,-1.15) {};
  \node[below right, font=\scriptsize, inner sep=1pt] at (g10D) {$D$};
  \node[bdry] (g10E) at (-0.95,-1.15) {};
  \node[below left, font=\scriptsize, inner sep=1pt] at (g10E) {$E$};
  \node[intv] (g10u) at (-0.45,0.05) {};
  \node[intv] (g10v) at (0.48,0.10) {};
  \draw[edge] (g10A) -- node[weight,pos=0.43] {3} (g10B);
  \draw[edge] (g10A)
    .. controls (-0.83,-0.57) and (0.31,-1.00) ..
    node[weight,pos=0.38] {2} (g10D);
  \draw[edge] (g10A) -- node[weight,pos=0.52] {4} (g10E);
  \draw[edge] (g10A) -- node[weight,pos=0.35] {5} (g10u);
  \draw[edge] (g10A) -- node[weight,pos=0.63] {5} (g10v);
  \draw[edge] (g10B) to[bend left=12] node[weight,pos=0.36] {3} (g10D);
  \draw[edge] (g10B) -- node[weight,pos=0.62] {5} (g10v);
  \draw[edge] (g10C) -- node[weight,pos=0.54] {5} (g10D);
  \draw[edge] (g10C) -- node[weight,pos=0.24] {1} (g10u);
  \draw[edge] (g10C) -- node[weight,pos=0.47] {3} (g10v);
  \draw[edge] (g10D) -- node[weight,pos=0.50] {5} (g10E);
  \draw[edge] (g10D) -- node[weight,pos=0.43] {4} (g10v);
  \draw[edge] (g10E) -- node[weight,pos=0.32] {3} (g10v);
  \draw[edge] (g10u) -- node[weight,pos=0.50] {3} (g10v);
\end{scope}

% Graph 11
\begin{scope}[shift={(4.00,-3.80)}]
  \node[graphlabel] at (0,1.95) {$G_{11}$};
  \node[bdry] (g11A) at (-1.40,1.05) {};
  \node[above left, font=\scriptsize, inner sep=1pt] at (g11A) {$A$};
  \node[bdry] (g11B) at (0.00,1.45) {};
  \node[above, font=\scriptsize, inner sep=1pt] at (g11B) {$B$};
  \node[bdry] (g11C) at (1.40,1.05) {};
  \node[above right, font=\scriptsize, inner sep=1pt] at (g11C) {$C$};
  \node[bdry] (g11D) at (0.95,-1.15) {};
  \node[below right, font=\scriptsize, inner sep=1pt] at (g11D) {$D$};
  \node[bdry] (g11E) at (-0.95,-1.15) {};
  \node[below left, font=\scriptsize, inner sep=1pt] at (g11E) {$E$};
  \node[intv] (g11u) at (-0.45,0.05) {};
  \node[intv] (g11v) at (0.48,0.10) {};
  \draw[edge] (g11A) -- node[weight,pos=0.52] {2} (g11E);
  \draw[edge] (g11A) -- node[weight,pos=0.48] {3} (g11v);
  \draw[edge] (g11B) to[bend left=13] node[weight,pos=0.55] {1} (g11D);
  \draw[edge] (g11B) to[bend right=13] node[weight,pos=0.55] {4} (g11E);
  \draw[edge] (g11B) -- node[weight,pos=0.38] {5} (g11u);
  \draw[edge] (g11B) -- node[weight,pos=0.62] {5} (g11v);
  \draw[edge] (g11C) -- node[weight,pos=0.55] {2} (g11D);
  \draw[edge] (g11C)
    .. controls (0.86,-0.59) and (-0.28,-1.01) ..
    node[weight,pos=0.44] {2} (g11E);
  \draw[edge] (g11C) -- node[weight,pos=0.32] {3} (g11u);
  \draw[edge] (g11C) -- node[weight,pos=0.58] {1} (g11v);
  \draw[edge] (g11D) -- node[weight,pos=0.50] {3} (g11E);
  \draw[edge] (g11D) -- node[weight,pos=0.44] {3} (g11v);
  \draw[edge] (g11E) -- node[weight,pos=0.44] {4} (g11v);
\end{scope}

% Graph 12
\begin{scope}[shift={(8.00,-3.80)}]
  \node[graphlabel] at (0,1.95) {$G_{12}$};
  \node[bdry] (g12A) at (-1.40,1.05) {};
  \node[above left, font=\scriptsize, inner sep=1pt] at (g12A) {$A$};
  \node[bdry] (g12B) at (0.00,1.45) {};
  \node[above, font=\scriptsize, inner sep=1pt] at (g12B) {$B$};
  \node[bdry] (g12C) at (1.40,1.05) {};
  \node[above right, font=\scriptsize, inner sep=1pt] at (g12C) {$C$};
  \node[bdry] (g12D) at (0.95,-1.15) {};
  \node[below right, font=\scriptsize, inner sep=1pt] at (g12D) {$D$};
  \node[bdry] (g12E) at (-0.95,-1.15) {};
  \node[below left, font=\scriptsize, inner sep=1pt] at (g12E) {$E$};
  \node[intv] (g12u) at (-0.45,0.05) {};
  \node[intv] (g12v) at (0.48,0.10) {};
  \draw[edge] (g12A) -- node[weight,pos=0.43] {5} (g12B);
  \draw[edge] (g12A)
    .. controls (-0.83,-0.57) and (0.31,-1.00) ..
    node[weight,pos=0.46] {3} (g12D);
  \draw[edge] (g12A) -- node[weight,pos=0.47] {2} (g12u);
  \draw[edge] (g12B) -- node[weight,pos=0.57] {5} (g12C);
  \draw[edge] (g12B) to[bend right=12] node[weight,pos=0.30] {5} (g12E);
  \draw[edge] (g12B) -- node[weight,pos=0.38] {4} (g12u);
  \draw[edge] (g12B) -- node[weight,pos=0.62] {4} (g12v);
  \draw[edge] (g12C)
    .. controls (0.86,-0.59) and (-0.28,-1.01) ..
    node[weight,pos=0.44] {4} (g12E);
  \draw[edge] (g12C) -- node[weight,pos=0.30] {2} (g12u);
  \draw[edge] (g12D) -- node[weight,pos=0.44] {4} (g12v);
  \draw[edge] (g12E) -- node[weight,pos=0.35] {2} (g12u);
  \draw[edge] (g12E) -- node[weight,pos=0.66] {1} (g12v);
  \draw[edge] (g12u) -- node[weight,pos=0.50] {3} (g12v);
\end{scope}
\end{tikzpicture}%
}
\caption{Graph representatives \(G_1,\ldots,G_{12}\) for the projected four-party balanced vectors in Table~\ref{tab:n5-top11-simplex}. Boundary terminals \(A,\ldots,E\) are labeled as solid dots; unlabeled solid dots are internal vertices. Edge labels denote their weights.}
\label{fig:n5-top11-graphs}
\end{figure}
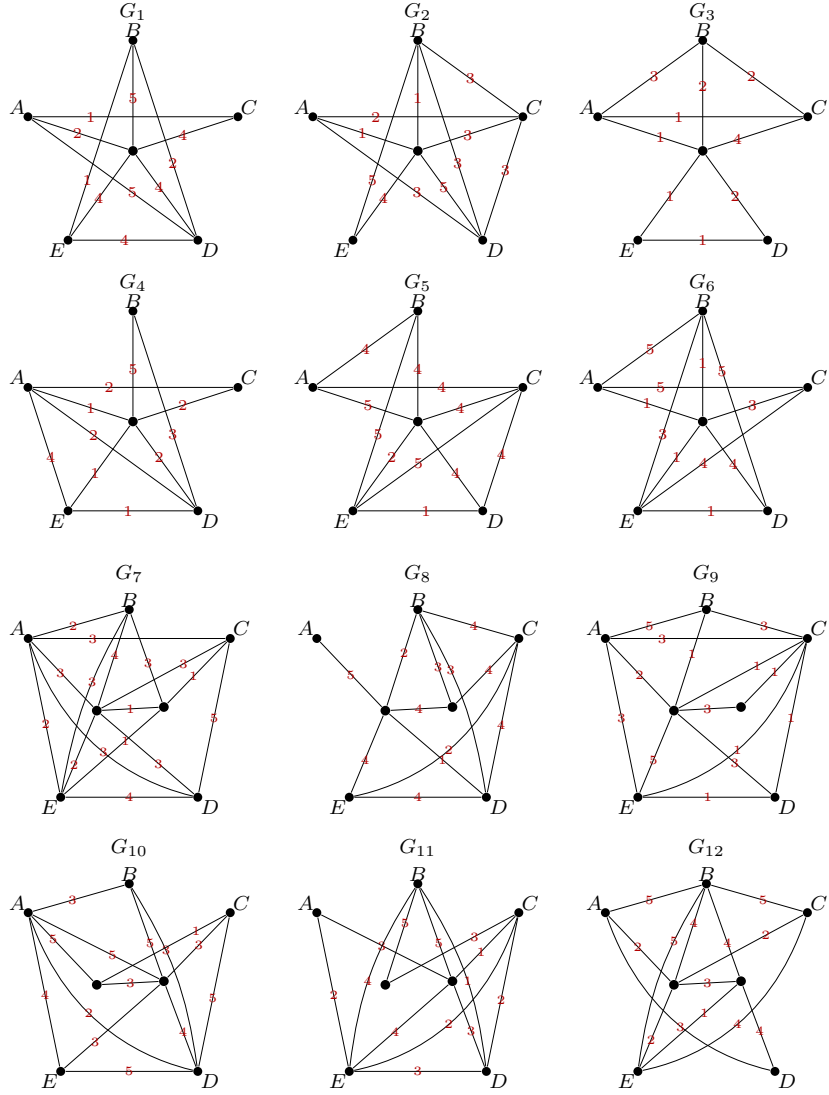

These twelve points form a simplex whose interior contains the origin. Indeed,
with
\begin{equation}
(n_1,\ldots,n_{12})
=
(75,293,212,873,19,266,11,70,191,19,3,186),
\end{equation}
one has
\begin{equation}
\sum_{i=1}^{12} n_i=2218,\qquad
\sum_{i=1}^{12} n_i y_i=0.
\end{equation}
The exact ranks are
\begin{equation}
\operatorname{rank}\{y_1,\ldots,y_{12}\}=11,\qquad
\operatorname{rank}\{(y_1,1),\ldots,(y_{12},1)\}=12,
\end{equation}
where the latter equation calculates the affine rank.
Thus the \(y_i\) are affinely independent in the eleven-dimensional projected
space, and the origin lies in the interior of the simplex with vertices
\(y_1,\ldots,y_{12}\). By the projected-simplex criterion,
\begin{equation}
V_{4\text{-bal}}^\perp
\cap
\operatorname{int}\!\left(C^{\rm HMEC}_5\right)
\neq
\varnothing .
\end{equation}
Equivalently, every nonzero signal in \(V_{4\text{-bal}}\) takes both signs on
holographic graph data. Therefore, no HMEI resides entirely in this eleven-dimensional, four-party balanced sector. {Note that two symmetric dimensions from this sector have already been discussed in \cite{Iizuka:2025caq}, while here we analyze the most general eleven-dimensional signal subspace.} This gives the \(n=5\)
evidence of the balanced-but-not-too-balanced conjecture.

To sum up, this conjecture is consistent with the known organization of HEIs in the
information basis~\cite{HubenyRangamaniRotaEntropyArrangement,HeHubenyRangamaniSuperbalance}. It is
also compatible with the $\HMEC_5$ evidence.
The finite HMEC obstruction suggests that holographically sign-definite entropy
inequalities are associated with intermediate-party balanceness.

\section{Conclusion}
\label{sec:conclusion}

We have developed the holographic multi-entropy cone as a refinement of the standard holographic entropy cone. In Section~\ref{sec:framework}, we established the general framework of the HMEC. In analogy with the HEC construction, we introduced holographic multi-entropy vectors, developed the corresponding weighted-graph model, and established the equivalence between graph models and time-reflection-symmetric bulk geometries. We further proved that the HMEC is a rational polyhedral cone and formulated the multicontraction map proof method, which provides systematic certificates for valid holographic multi-entropy inequalities.

Section~\ref{sec:low-dimensional} applies this HMEC framework to the first
nontrivial cases. In the total-label convention, with the purifier included
among the \(n\) labels, we computed the complete cones \(\HMEC_3\) and
\(\HMEC_4\). The former cone has \(4\) concrete facets in $2$ orbits and \(4\) extreme rays in $2$ orbits, while the latter has \(27\) facets in $7$ orbits and \(49\) extreme rays in $9$ orbits. All of these above extreme rays admit explicit
graph realizations. Section~\ref{sec:low-dimensional} also explained why a
full \(\HMEC_5\) facet enumeration is already a substantially harder problem:
its natural superbalanced sector where HMEIs are searched is comparable in dimension to the
corresponding sector of the standard six-region HEC problem, while the graph
realization of extreme rays must also match the full partition-labeled multiway-cut data.

In Section~\ref{sec:conjectures}, two guiding conjectures are formulated suggested by
these results. The first is that standard HEC facets arise from convex
combinations of HMEC facets in which the multi-entropy components cancel. The second is the
balanced-but-not-too-balanced principle: HMEC facets should occupy
intermediate-party balance classes rather than being arbitrarily high party balanced. The complete \(n=3,4\) data and the partial \(n=5\) results
are consistent with both conjectures.

The main next steps are to prove or refine these conjectures and to construct
systematic series of HMEIs. We have already found two such series, suggesting
that the low-dimensional facets computed in this paper are the first examples of a
larger structure. We also plan to use the Hasse diagram of Sperner hypergraphs
to study how linear combinations of entropies and multi-entropies
probe multipartite entanglement. This perspective will be developed in our
forthcoming work
\cite{JuSunZhaoResolutionSignal}.

\section*{Acknowledgements}

We thank Bart{\l}omiej Czech, Yichen Feng, Jonathan Harper, Ya-Wen Sun, and
Tadashi Takayanagi for useful discussions.
Part of this work was completed while X.-X.~Ju was visiting the Yukawa
Institute for Theoretical Physics (YITP), Kyoto University, under the Young
International Researcher Invitation Program. X.-X.~Ju thanks YITP for its
hospitality.
This work was supported by the National Natural Science Foundation of China
under Grant No. 12575068.

\bibliographystyle{JHEP}
\bibliography{refs}

\end{document}